%% file: Toeplitz.tex
\documentclass[a4paper,10pt]{article}

\usepackage[utf8]{inputenc} % encoding
\usepackage{color}
\usepackage[ruled, vlined]{algorithm2e}
\usepackage{relsize}
\usepackage{bbold}

\input{preamble} % preamble
\input{math} % math macros
\graphicspath{{images/}} % graphics path

\usepackage[normalem]{ulem}
\title{\bfseries Recovery of Binary Sparse Signals from Structured Biased Measurements}
\author{\hspace*{-1cm} Sandra Keiper\\[.5em]
  \small{\textsc{\hspace*{-1cm}Institut für Mathematik, Technische Universit\"at Berlin }}\\[.5em]
}

\begin{document}

%%% todos
\listoftodos
%%% content macros
\input{content-macros}
%%% title
\maketitle
%%% abstract

\begin{abstract}
In this paper we study the reconstruction of binary sparse signals from partial random circulant measurements. We show that the reconstruction via the least-squares algorithm is as good as the reconstruction via the usually used program basis pursuit. We further show that we need as many measurements to recover an $s$-sparse signal $x_0\in\R^N$ as we need to recover a dense signal, more-precisely an $N-s$-sparse signal $x_0\in\R^N$. We further establish stability with respect to noisy measurements.
\end{abstract}

\vspace{.1in}
\vspace{.1in}

% \pagebreak
%%% content
\input{Intro}
\input{Symmetry} \label{sec:Symmetry}

\input{Numerics} \label{sec:Num}
\input{UpperBound}	\label{sec:UpBound}

\subsection*{Acknowledgements}
Sandra Keiper acknowledges support by the DFG Collaborative Research Center TRR 109 “Discretization in Geometry and Dynamics” and support by the Berlin Mathematical School.

%%% references
%\nocite{}

\bibliographystyle{amsplain}
\bibliography{Toeplitz}

\end{document}

%% file: preamble.tex
\makeatletter
\newcommand{\LoadPackagesNow}{}
\newcommand{\LoadPackageLater}[1]{%
   \g@addto@macro{\LoadPackagesNow}{%
      \usepackage{#1}%
   }%
}
\makeatother

\usepackage{xspace}
\usepackage{xargs}
\usepackage{ifthen}
\usepackage{etoolbox}

\usepackage[
	table % Load for using rowcolors command in tables
]{xcolor}
\usepackage[pdftex%
]{graphicx}
\usepackage{float}
\usepackage{wrapfig}

\usepackage{subfigure}
\usepackage{caption}
\usepackage{multicol}
\usepackage[
   tbtags,    % 'Top-or-bottom tags': For a split equation, place equation numbers level
   sumlimits,  %(default) Place the subscripts and superscripts of summation
   nointlimits, % (default) Opposite of intlimits.
   namelimits, % (default) Like sumlimits, but for certain 'operator names' such as
   reqno,     % Place equation numbers on the right.
]{amsmath} %

\usepackage{amsfonts}
\usepackage{mathrsfs} %% Schreibschriftbuchstaben für den Mathematiksatz (nur Großbuchstaben)
\usepackage{dsfont}
\usepackage{amssymb}
\LoadPackageLater{amsthm}
\LoadPackageLater{thmtools}
\usepackage[fixamsmath,disallowspaces]{mathtools}
\mathtoolsset{showonlyrefs}
\usepackage{geometry}
\usepackage[%
	square,	% for square brackets;
	comma,	% to use commas as separaters;
	numbers,	% for numerical citations;
	sort,		% orders multiple citations into the sequence in which they appear in the list of eferences;
	sort&compress,    % as sort but in addition multiple numerical citations
]{natbib}

\usepackage{framed}
\usepackage[normalem]{ulem}      % Zum Unterstreichen
\usepackage{soul}		            % Unterstreichen, Sperren
\usepackage{url} % Setzen von URLs. In Verbindung mit hyperref sind diese auch aktive Links.
\usepackage{lipsum}
\usepackage[textsize=tiny,english,colorinlistoftodos,disable]{todonotes}
\usepackage{enumitem}
\definecolor{pdfurlcolor}{rgb}{0,0,0.6}
\definecolor{pdffilecolor}{rgb}{0.7,0,0}
\definecolor{pdflinkcolor}{rgb}{0,0,0.6}
\definecolor{pdfcitecolor}{rgb}{0,0,0.6}
\usepackage[
  colorlinks=true,         % Links erhalten Farben statt Kaeten
  urlcolor=pdfurlcolor,    % \href{...}{...} external (URL)
  filecolor=pdffilecolor,  % \href{...} local file
  linkcolor=pdflinkcolor,  %\ref{...} and \pageref{...}
  citecolor=pdfcitecolor,  %
  raiselinks=true,			 % calculate real height of the link
  breaklinks,              % Links berstehen Zeilenumbruch
  verbose,
  hyperindex=true,         % backlinkex index
  linktocpage=true,        % Inhaltsverzeichnis verlinkt Seiten
  hyperfootnotes=false,     % Keine Links auf Fussnoten
  bookmarks=true,          % Erzeugung von Bookmarks fuer PDF-Viewer
  bookmarksopenlevel=1,    % Gliederungstiefe der Bookmarks
  bookmarksopen=true,      % Expandierte Untermenues in Bookmarks
  bookmarksnumbered=true,  % Nummerierung der Bookmarks
  bookmarkstype=toc,       % Art der Verzeichnisses
  plainpages=false,        % Anchors even on plain pages ?
  pageanchor=true,         % Pages are linkable
  pdftitle={Compressed Sensing for Finite-Valued Signals},             % Titel
  pdfauthor={Sandra Keiper},            % Autor
  pdfcreator={LaTeX, hyperref, KOMA-Script}, % Ersteller
  pdfdisplaydoctitle=true, % Dokumententitel statt Dateiname im Fenstertitel
  pdfstartview=FitH,       % Dokument wird Fit Width geaefnet
  pdfpagemode=UseOutlines, % Bookmarks im Viewer anzeigen
  pdfpagelabels=true,           % set PDF page labels
  pdfpagelayout=SinglePage, % einseitige Darstellung
]{hyperref}

\LoadPackagesNow

\usepackage{bm}
\makeatletter
\g@addto@macro\bfseries{\boldmath}
\makeatother

\newcommand{\ifargdef}[3][{}]{\ifthenelse{\equal{#2}{}}{#1}{#3}}

%% file: math.tex
\newtheoremstyle{claim}
	{\topsep}{\topsep}%
	{\itshape}%         Body font
	{}%         Indent amount (empty = no indent, \parindent = para indent)
	{\bfseries\boldmath}% Thm head font
	{}%        Punctuation after thm head
	{.5em}%     Space after thm head (\newline = linebreak)
	{\thmname{#1} \thmnumber{#2} \thmnote{(#3)}}%         Thm head spec
\newtheoremstyle{definition}
	{\topsep}{\topsep}%
	{}%         Body font
	{}%         Indent amount (empty = no indent, \parindent = para indent)
	{}% Thm head font
	{}%        Punctuation after thm head
	{.5em}%     Space after thm head (\newline = linebreak)
	{{\bfseries\thmname{#1} \thmnumber{#2}} \thmnote{(#3)}}%         Thm head spec
\newtheoremstyle{remark}
	{\topsep}{\topsep}%
	{}%         Body font
	{}%         Indent amount (empty = no indent, \parindent = para indent)
	{}% Thm head font
	{}%        Punctuation after thm head
	{.5em}%     Space after thm head (\newline = linebreak)
	{{\itshape\thmname{#1}:}}%         Thm head spec
\newtheoremstyle{example}
	{\topsep}{\topsep}%
	{}%         Body font
	{}%         Indent amount (empty = no indent, \parindent = para indent)
	{}% Thm head font
	{}%        Punctuation after thm head
	{.5em}%     Space after thm head (\newline = linebreak)
	{{\itshape\thmname{#1}:}}%         Thm head spec

\declaretheorem[style=claim,numberwithin=section]{theorem}
\declaretheorem[style=claim,sibling=theorem]{proposition}
\declaretheorem[style=claim,sibling=theorem]{remark}

\declaretheorem[style=definition,sibling=theorem]{definition}
\newcommand{\opleft}[1]{\mathopen{}\left#1}
\newcommand{\opright}[1]{\right#1\mathclose{}}
\newcommandx{\braces}[4]{%
\ifstrequal{#3}{normal}{#1#4#2}{%
\ifstrequal{#3}{auto}{\left#1#4\right#2}{%
\ifstrequal{#3}{opauto}{\opleft#1#4\opright#2}{%
#3#1#4#3#2}}}%
}
\newcommandx{\opannot}[3][3=\downarrow]{\stackrel{\mathclap{\substack{#1 \\ #3 \vspace{2pt}}}}{#2}}
\newcommandx{\lineannot}[3][3=\rightarrow]{\mathllap{\boxed{\text{\textsmaller{#1}}} #3} #2}
\newcommandx{\multilineannot}[4][4=\rightarrow]{\mathllap{\boxed{\parbox{#1}{\RaggedRight\textsmaller{#2}}} #4} #3}
\newcommand{\R}{\mathbb{R}} % real numbers
\newcommand{\C}{\mathbb{C}} % complex numbers
\newcommand{\E}{\mathbb{E}} 
\newcommand{\sprod}[1]{\left\langle #1 \right\rangle}

\DeclareMathOperator{\one}{\mathbb{1}}

\newcommand{\prb}[1]{\mathbb{P}\left( #1 \right)}
\newcommand{\erw}[1]{\mathbb{E}\left( #1 \right)}

\newcommand{\suchthat}[1][normal]{\ifstrequal{#1}{normal}{\mid}{#1|}} % seperator in sets (#1op = size)

\newcommandx{\intvcl}[3][1=normal]{\braces{[}{]}{#1}{#2, #3}} % closed interval (#1op=size, #2=left bound, #3=right bound)
\newcommandx{\intvop}[3][1=normal]{\braces{(}{)}{#1}{#2, #3}} % open interval
\newcommandx{\intvclop}[3][1=normal]{\braces{[}{)}{#1}{#2, #3}} % half-open interval (right)
\newcommandx{\intvopcl}[3][1=normal]{\braces{(}{]}{#1}{#2, #3}} % half-open interval (left)

\DeclareMathOperator*{\sign}{sign} % argmax
\newcommandx{\abs}[2][1=normal]{\braces{\lvert}{\rvert}{#1}{#2}} % absolute value
\newcommandx{\ceil}[2][1=normal]{\braces{\lceil}{\rceil}{#1}{#2}} % ceil
\newcommandx{\floor}[2][1=normal]{\braces{\lfloor}{\rfloor}{#1}{#2}} % floor
\newcommandx{\round}[2][1=normal]{\braces{[}{]}{#1}{#2}} % round
\newcommandx{\der}[1]{D^{#1}} % differential operator (#1 = multiindex)
\newcommandx{\partder}[4][1={},4={}]{\frac{\partial^{#4} #2}{\partial #3^{#4}}\ifargdef{#1}{\Big|_{#1}}} % partial derivative (#1=point of evaluation, #2=function, #3=variable, #4=order)
\newcommandx{\integ}[4][1={},2={}]{\int_{#1}^{#2} #3 \, #4} % integral (#1op=lower bound, #2op=upper bound, #3=integrand, #4=differential form)
\newcommandx{\asympffaster}[2][1=normal]{o\braces{(}{)}{#1}{#2}} % asymptotically faster (proper) (#1op=size)
\newcommandx{\asympfaster}[2][1=normal]{O\braces{(}{)}{#1}{#2}} % asymptotically faster
\newcommandx{\asympeq}[2][1=normal]{\Theta\braces{(}{)}{#1}{#2}} % asymptotically equal
\newcommandx{\asympsslower}[2][1=normal]{\omega\braces{(}{)}{#1}{#2}} % asymptotically slower (proper)
\newcommandx{\asympslower}[2][1=normal]{\Omega\braces{(}{)}{#1}{#2}} % asymptotically slower
\makeatletter
\def\imod#1{\allowbreak\mkern1mu({\operator@font mod}\,\,#1)}
\makeatother

\newcommandx{\norm}[2][1=normal]{\braces{\|}{\|}{#1}{#2}} % norm
\renewcommandx{\sp}[3][1=normal]{\braces{\langle}{\rangle}{#1}{#2, #3}} % inner product (#1op=size, #2=left, #3=right)
\newcommand{\adj}[1]{{#1}^\ast} % adjoint operator
\newcommandx{\End}[2][2={}]{\mathcal{L}\opleft( #1 \ifargdef{#2}{, #2} \opright)} % endomorphism (#1=from, #2op=to)
\newcommandx{\measure}[2][1=normal]{\operatorname{vol}\braces{(}{)}{#1}{#2}} % Lebesgue-measure/volume of a set
\DeclareMathOperator{\supp}{supp} % support
\newcommandx{\Leb}[3][1={},3=normal]{L^{#2}\ifargdef{#1}{\braces{(}{)}{#3}{#1}}{}} % Lebesgue spaces (#1op=set, #2=exponent)
\newcommandx{\Lebnorm}[4][1=normal,3={2},4={}]{\norm[#1]{#2}_{\Leb[#4]{#3}}} % Lebesgue norm (#1op=size, #2=content, #3op=exponent, #4op=set)
\renewcommandx{\l}[3][1={},3=normal]{\ell^{#2}\ifargdef{#1}{\braces{(}{)}{#3}{#1}}} % lp sequence spaces (#1op=set, #2=exponent)
\newcommandx{\lnorm}[4][1=normal,3={2},4={}]{\norm[#1]{#2}_{\l[#4]{#3}}} % lp norm (#1op=size, #2=content, #3op=exponent, #4op=set)
\newcommandx{\Smooth}[4][1={},3={},4=normal]{C_{#3}^{#2}\ifargdef{#1}{\braces{(}{)}{#4}{#1}}} % space of differentiable functions (#1op=set, #2=order, #3op=modifier)
\newcommandx{\Schwartz}[2][1={},2=normal]{\mathscr{S}\ifargdef{#1}{\braces{(}{)}{#2}{#1}}} % space of Schwartz functions
\newcommandx{\Schwartzpoly}[2][1=normal]{\braces{\langle}{\rangle}{#1}{\abs[#1]{#2}} } % Schwartz polynomial
\newcommandx{\Tempdistr}[2][1={},2=normal]{\mathscr{S}'\ifargdef{#1}{\braces{(}{)}{#2}{#1}}} % tempered distributions
\newcommandx{\distrinp}[3][1=normal]{\braces{\langle}{\rangle}{#1}{#2, #3}} % evaluation of a tempered distribution (#1op=size, #2=distribution, #3=Schwartz function)
\newcommand{\Linedistr}[1][]{\mathfrak{L}\ifargdef{#1}{_{#1}}{}} % line distribution
\newcommandx{\ft}[3][1=default,2=auto]{
\ifstrequal{#1}{default}{\widehat{#3}}{
\ifstrequal{#1}{long}{{\braces{(}{)}{#2}{#3}}^{\wedge}}{}}} % Fourier transform (hat-notation) (#1op=long expression mode, #2op=size, #3=content)
\DeclareMathOperator{\tr}{tr}

\newcommandx{\ift}[3][1=default,2=auto]{
\ifstrequal{#1}{default}{\check{#3}}{
\ifstrequal{#1}{long}{{\braces{(}{)}{#2}{#3}}^{\vee}}{}}} % inverse Fourier transform (hat-notation) (#1op=long expression mode, #2op=size, #3=content)

%% file: content-macros.tex
%%% --> Basic notions <--
\newcommand{\OpAnalysis}[1]{T_{#1}} % analysis operator
\newcommand{\OpSynthesis}[1]{\adj T_{#1}} % synthesis operator
\newcommand{\OpFrame}[1]{S_{#1}} % frame operator
\newcommand{\defsf}{\varphi} % default Schwartz function
%%% --> Abstract Model <--
\newcommand{\InpSp}{\mathcal{H}} % inpainting space
\newcommand{\InpSpK}{\InpSp_K} % known Space
\newcommand{\ProK}{P_K} % projection on the known space
\newcommand{\InpSpM}{\InpSp_M} % missing space
\newcommand{\ProM}{P_M} % projection on the missing space
\newcommand{\sig}{x^0} % original, undamaged signal
\newcommand{\sigrec}{x^\star} % recovered signal
\newcommand{\PF}{\Phi} % default Parseval frame
\newcommand{\pf}{\phi} % default element of this Parseval frame
%%% --> Analysis Tools <--
\newcommand{\cluster}{\Lambda} % default cluster
\newcommand{\concentr}[2]{\kappa\ifargdef{#1}{\opleft( #1, #2 \opright)}} % concentration
\newcommand{\clustercoh}[2]{\mu_c \ifargdef{#1}{( #1 , #2)}} % cluster coherence (#1 = cluster, #2 = frame)
\newcommand{\anorm}[2]{\norm{#1}_{1,#2}} % l1-analysis norm (#2 = frame)
%%% --> Shearlets <--
\newcommand{\ver}{\mathrm{v}} % vertical component
\newcommand{\hor}{\mathrm{h}} % horizontal component
\newcommand{\dir}{\imath} % variable component
\newcommand{\meyerscal}{\phi} % Meyer scaling function
\newcommand{\Scalfunc}{\Phi} % coarse scaling function
\newcommand{\Corofunc}{W} % corona scaling function
\newcommand{\Coro}{\mathscr{K}} % frequency corona
\newcommand{\conefunc}{v} % generator of the cone function
\newcommand{\Conefunc}[1]{V_{(#1)}} % cone function
\newcommand{\Cone}[1]{\mathscr{C}_{(#1)}} % cone (#1 = ver/hor)
\newcommand{\pscal}{A} % scaling matrix
\newcommand{\pshear}{S} % shearing matrix
\newcommand{\pscalcone}[2]{A_{#1,(#2)}} % scaling matrix (#1 = alpha, #2 = ver/hor)
\newcommand{\shearcone}[1]{S_{(#1)}} % shearing matrix (#1 = shearing number)
\newcommand{\unishplain}{\psi} % universal shearlet element (without indices)
\newcommand{\unishplainft}{\ft{\unishplain}} % Fourier transformed universal shearlet element (without indices)
\newcommand{\unish}[5][{}]{\unishplain_{#3,#4,#5}^{#2\ifargdef{#1}{,(#1)}}} % universal shearlet element (#1=direction, #2=alpha, #3=scale, #4=shear, #5=translation)
\newcommand{\unishft}[5][{}]{\unishplainft_{#3,#4,#5}^{#2\ifargdef{#1}{,(#1)}}} % Fourier transformed universal shearlet element 
\newcommand{\Scalparamdomain}{\mathsf{A}} % discretization set of scaling parameters
\newcommand{\aj}{\alpha_j} % shortcut for scaling sequence
\newcommandx{\Unish}[3][1=\meyerscal,2=\conefunc,3=(\aj)_j]{\operatorname{SH}(#1, #2, #3)} % universal shearlet system (#1=Meyer scaling generator, #2=cone generator, #3=scaling sequence)
\newcommandx{\UnishLow}[1][1=\meyerscal]{\operatorname{SH}_{\mathrm{Low}}(#1)} % coarse scaling system (#1=Meyer scaling generator)
\newcommandx{\UnishInt}[3][1=\meyerscal,2=\conefunc,3=(\aj)_j]{\operatorname{SH}_{\mathrm{Int}}(#1, #2, #3)} % interior shearlet system
\newcommandx{\UnishBound}[3][1=\meyerscal,2=\conefunc,3=(\aj)_j]{\operatorname{SH}_{\mathrm{Bound}}(#1, #2, #3)} % boundary shearlet system
\newcommand{\Unishgroup}{\Gamma} % index set for \Unish
\newcommand{\Unishind}{\gamma} % index for some element of \Unish
\newcommand{\lmax}{l_j} % shortcut for maximal shearing number at scale j
%%% --> Model <--
\newcommand{\weight}{w} % weighting function
\newcommand{\wlen}{\rho} % length of the singularity
\newcommand{\model}{{\weight\Linedistr}} % line model
\newcommand{\modelrec}{\model^\star} % recovered model
\newcommand{\Corofilter}{F} % frequency filter
\newcommand{\mdiam}{h} % width of the mask
\newcommand{\mask}[1]{\mathscr{M}_{#1}} % mask set
%%% --> Image-Inpainting <--
\newcommand{\Unishshort}{\Psi} % shortcut for \Unish
\newcommand{\scalpm}[1]{#1^{\pm1}} % notation for neighboring technicality (#1=some set)
\newcommand{\translind}[2]{#1^{(#2)}} % directional index (#1=something, #2=ver/hor)

%% file: Intro.tex
\section{Introduction}
A recent mathematical framework that ensures recovery of sparse vectors from incomplete information is \emph{Compressed sensing}. In this context, incomplete information  refers to the fact that linear systems of the form
$$
Ax_0=y
$$ 
can only be solved uniquely for general $x_0\in \R^N$ if $A\in \R^{M,N}$ is quadratic and invertible, i.e., if $M<N$ the information is incomplete. By imposing an a-priori structure on $x_0$, however, the ill-posed problem for $M<N$ can be turned into a well-posed one. An important structure of $x_0$ is that of \emph{sparsity}, which means that only a few entries of $x_0$ are different from zero. Another relevant assumption on $x_0$ is that its entries stem from a finite alphabet. In this paper we will study sparse signals whose entries stem from a binary alphabet, e.g., $x_0\in \{0,1\}^N$. Such signals appear for example in wireless communications, where the transmitted signals are sequences of bits. Moreover, in certain types of communication networks it is appropriate to assume that only a few transmitters are active at a certain instance, which naturally induces sparsity.

Note that binary signals are in particular symmetric in the sense, that if $x_0$ is a dense signal, $\mathds{1}-x_0$ is sparse. In the recent publication \cite{FlKei}, we have proven that using shifted random matrices such signals can be recovered from a particularly small number of measurements and that this number reflects the mentioned symmetry of $x_0$. This means that we need the same number of measurements to recover an $s$-sparse signal as we need to recover an $(N-s)$-sparse signal.

The considered measurement matrices, however, are of somewhat limited use in applications. The reasons are diverse. Often the design of the measurement matrix is given by the applications with little or even  
no freedom to design it. Moreover, unstructured matrices, such as random or particularly Gaussian and Rademacher matrices, do not allow for a fast matrix multiplication, which may speed up recovery algorithms significantly. Beyond that storing a large unstructured matrix might be difficult. Hence, from a computational and application-technological point of view it would be desirable to use structured random matrices. 

Up until now there are only few good recovery conditions for completely deterministic measurement matrices available. One, therefore, should allow for some randomness to come into play. In this work, we will consider biased partial random circulant matrices for the measurement process. A precise definition of such matrices will be given in Subsection \ref{sec:IntMainRes} (Equation \eqref{eqn:BiasedCirc}). A main difference is that those matrices depend on only one random vector and accordingly $N$ random variables, whereas (sub-) Gaussians either depend on $MN$ random variables (following the definition in \cite{FR}) or on $M$ random row vectors (following the definition in \cite{Tropp}).

The concern of this work is to proof recovery guarantees for binary, sparse signals from biased random partial circulant and Toeplitz measurements.

Partial random circulant matrices (centered and not biased) have already been successfully applied to the classical compressed problem (see for example \cite{StrctRau}) and in one-bit compressed sensing (\cite{DJR,DS}).
In \cite{DJR}, \cite{DS} and \cite{FKS}, the reconstruction of sparse signals from binary Gaussian circulant measurements was also considered. The main difference is that in those papers the measurements $y=\sign(Ax_0)$ are assumed to be binary whereas we assume that the signal $x_0$ itself is binary. Besides, the proof techniques are very different.

\subsection{Preliminaries}
To put our results in a precise setting we first aim to introduce some notation. We define $[N]:=\left\{1,\dots, N\right\}$ and denote the standard unit vectors with $e_i$ for $i\in [N]$, i.e., the vector which is zero everywhere except at the $i$-th entry it is equal to one. Further, we denote with $\one$ the matrix or vector, respectively, which is equal to one in each entry. For a subset $S\subset [N]$ and a vector $x=(x_1,\dots,x_N)=(x(1),\dots, x(N))\in \R^N$ the notation $x_S$ refers to the following vector
$$x_S(i)=\begin{cases}
x_i, & \text{if}\quad i\in S\\ 0, & \text{else}.
\end{cases}$$
Further $\|\cdot\|_0$ denotes the $\ell_0$-norm and $\|\cdot\|_p$ the $\ell_p$-norm for $p>0$, i.e., for $x=[x_1,\dots,x_N]\in \R^N$
$$\|x\|_0:=\abs{\left\{i: x_i\neq 0\right\}} \qquad \text{and} \qquad \|x\|_p^p:=\sum_{i=1}^N \abs{x_i}^p.$$

For a matrix $A\in \R^{M,N}$ the adjoint matrix is denoted by $A^*$. The notation $\supp(x)$ refers to the support of a vector $x\in\R^N$, i.e., to the set of non-zero entries of $x$.

For two real numbers $a,b\in \R$ we write $a\gtrsim b$ if there exists a constant $c>0$, which is independent from $a$ and $b$, such that $a\ge c b$.

Let $\Theta\subset [N]$ then we define $$i+\Theta:=\{i+k\imod{N}: k\in \Theta\}.$$
Moreover, for a linear operator $L:\R^N\to \R^M$ with matrix representation $L=[L_{i,j}]_{i,j=1}^{MN}$ the Hilbert-Schmidt norm of $L$ is denoted by $$\norm{L}_{HS}:=\sqrt{\tr(L^*L)}=\sqrt{\sum_{i=1}^M\sum_{j=1}^NL_{i,j}^2},$$ and $$\norm{L}:= \sup_{\norm{w}_{2} \leq 1} \norm{Lw}_{2}$$ is the operator norm of $L$. Note that the operator norm of $L:\R^N\to \R^N$ corresponds to the largest eigenvalue in absolute of $L$. Further note that for $A,B\in \R^{N,N}$ with $B=[b_1,\dots,b_N]$ it holds true that
\begin{align}\label{eqn:EstHS}
\|AB\|_{\text{HS}}^2=\sum_{i=1}^N\|Ab_i\|_2^2\le\sum_{i=1}^N \|A\|^2\|b_i\|_2^2=\|A\|\|B\|^2_{\text{HS}},
\end{align}
where we used the consistency of the operator norm with the Euclidean norm $\|\cdot\|_2$.

For $A\in \R^{M,N}$ and $S\subset[N]$ we let $A_S$ be either the matrix which consists of the rows of $A$ corresponding to the indices in $S$, i.e. $A_S\in \R^{s,N}$, or the matrix whose rows corresponding to the indices in $S$ equal those of $A$ and all other columns are zero, i.e., $A_S\in \R^{M,N}$. We further let $A^S\in \R^{M,s}$ or $A^S\in \R^{M,N}$ be the matrix whose columns corresponding to $S$ are deleted or substituted with zero-columns.

Finally, we aim to recall Gershgorin circle theorem, which will be an important tool for our proofs.

\begin{theorem}[Gershgorin circle theorem \cite{Gersch}]\label{thm:Gers}
Let $A=[a_{i,j}]_{i,j=1}^N\in \C^{N,N}$ and for $i\in[N]$ let $R_{i}=\sum _{j\neq {i}}\left|a_{ij}\right|$ be the sum of the absolute values of the non-diagonal entries in the $i$-th row.  Further define $D(a_{ii},R_{i})\subseteq \mathbb {C}$ be a closed disc centered at $a_{ii}$ with radius $R_{i}$. Every eigenvalue of $A$ lies then within at least one of the \emph{Gershgorin discs} $D(a_{ii},R_{i})$.
\end{theorem}

\subsection{Reconstruction of Binary Signals}
There are several compressed sensing approaches for the reconstruction of nonnegative sparse signals from random measurements \cite{Stojnic+,DT,JK}. As binary vectors are particularly nonnegative, those approaches can readily be applied to binary vectors. Let us therefore shortly review one of the approaches for the recovery of nonnegative signals.

It has become evident that basis pursuit restricted to the positive orthant $$\R^N_+:=\left\{x=(x_i)_{i=1}^N\in \R^N: x_i \ge 0, i\in [N]\right\}$$ has a strong performance at recovering nonnegative-valued sparse signals $x_0$ from the measurements $y=Ax_0$. This is the following program:
\begin{align}\label{P+}
\min\|x\|_1 \quad \text{subject to} \quad Ax=y \quad \text{and} \quad x\in \R^N_+,
\end{align} 

%Wheres the results of Stojnic are based on some properties of the nullspace of the measurement matrix $A$, Donoho and Tanner presented are more geometric approach to the problem. Their argument relies on the fact that if $F$ is defined as the convex hull of the vectors $\set{e_i: i\in K}$, it holds: A certain nonnegative-valued signal $x_0$ having non-zero entries on a set $K$ is recovered by \eqref{P+} if and only if $AF$ is a $k$-face of the projected polytope $A\R^N_+$. They moreover computed the probability of such a face ''surviving'' the projection with a random matrix $A$ which fulfills some specific properties.
%
%Having introduced these notions, we can recall a specific result (Lemmas 2.2, 2.3) from \cite{DT} which we will also need later on. It claims the following: Let $F$ denote a $k$-face of the orthant $\R^N_+$ and let $A\in \R^{m,N}$, with $m<N$, be in general position and have an orthant symmetric nullspace. Then
%	\begin{align}\label{eqn:survivingFaces}
%	\prb{AF \text{ is a } k-\text{face of } A\R^N_+}=1-P_{N-m, N-k},
%	\end{align}
%	where
%	\begin{align}
%	P_{ij}=2^{-j+1}\sum_{l=0}^{i-1}{i-1\choose l}, \quad {i,j\in \N}.
%	\end{align}
%
%This result in particular gives the probability of success of the program \eqref{P+} recovering a certain $k$-sparse nonnegative-valued vector $x_0 \in \R^N$.

Even so the mentioned approach is applicable to binary signals, it is already known that there are methods which yield even stronger recovery guarantees for binary signals. The canonical approach is to use the following adaptation of basis pursuit, to which one typically refers to as \emph{box-constrained basis pursuit}:
\begin{align}\label{Pbin}
\min\|x\|_1 \quad \text{subject to} \quad Ax=y \quad \text{and} \quad x\in [0,1]^N.
\end{align}
In \cite{Stojnic_Binary, KKLP}. the following equivalent condition for the success of \eqref{Pbin} has been shown. The vector $\one_S$, $S\subset [N]$, is the unique solution of \eqref{Pbin} if and only if
\begin{align}
\ker(A)\cap N^+\cap H_S=\{0\},
\end{align}
where $\ker(A)$ denote the nullspace of $A$,
\begin{align}
N^+:=\left\{w\in \R^N: \sum_{i=1}^N w_i\le 0\right\} \quad \text{and} \quad H_S:=\left\{w\in \R^N: w_i\le 0\text{ for } i \in S \text{ and } w_i\ge 0 \text{ for } i \notin S\right\}.
\end{align}

Least-squares on the other hand is an algorithm which is usually not well-applicable to sparse vectors since the $\ell_2$-norm does not promote sparsity. In \cite{FlKei}, however, it has been shown that least-squares with box-constraints works comparably well for binary-valued sparse signals, when reconstructing from biased sub-Gaussian measurements, even, if the measurements are contaminated with noise. \emph{Least-squares with box-constraints} is the following program
\begin{align}\label{LSBin}
\min\|Ax-y\|_2 \quad \text{subject to} \quad x\in [0,1]^N.
\end{align}
Note, that this fact has some practical impact. On the one hand least-squares is less complex and on the other hand it ensures a priori robustness in case of noisy measurements.

\subsection{Biased Random Matrices}
In compressed sensing the measurement matrix is often assumed to be (sub-) Gaussian, meaning that each entry of the measurement matrix $A\in\R^{M,N}$ is an independent copy of some (sub-) Gaussian. More precisely, we call  $A\in\R^{M,N}$ \emph{Gaussian}, if its entries are independently drawn from a renormalized normal distribution, i.e.,
\begin{align*}
A=\frac{1}{\sqrt{M}}\left[a_{i,j}\right]_{i,j=1}^{M,N} \qquad \text{with} \qquad a_{i,j}\sim \mathcal{N}(0,1).
\end{align*} 
A more general type of measurement matrices are sub-Gaussians, whose entries follow a sub-Gaussian distribution:
\begin{definition}
Let $(\Omega, \Sigma, \mathbb{P})$ be a probability space. Further let $X:\Omega\to \R$ be a random variable. The \emph{sub-Gaussian} or \emph{Orclitz-2}-norm of $X$ is given by
\begin{align}
\|X\|_{\Psi_2} = \sup_{p\ge 1} p^{-\frac{1}{2}}\E\left(|X|^p\right)^{\frac{1}{p}}.
\end{align}	
We call $X$ \emph{sub-Gaussian} if $\|X\|_{\Psi_2}<\infty$.
\end{definition}
A particular example for a sub-Gaussian matrix is a Rademacher matrix.  $A\in\R^{M,N}$ is called \emph{Rademacher matrix}, if its entries follow a Rademacher distribution, i.e., are chosen to be $-1$ or $1$ with equal probability:
\begin{align}
A=\frac{1}{\sqrt{M}}\left[a_{i,j}\right]_{i,j=1}^{M,N} \qquad \text{with} \qquad \mathbb{P}(a_{i,j}=1)=\mathbb{P}(a_{i,j}=-1)=\frac{1}{2}.
\end{align} 
Rademacher variables $X$ are indeed sub-Gaussians with norm $\|X\|_{\Psi_2}=1$, because $|X|=1$. Further, note that sub-Gaussians are sometimes also defined by assuming that the rows are independent random vectors, which fulfill some specific properties such as sub-Gaussian marginals (cf. \cite{Tropp}). This basically is some generalization of the definition used in the underlying paper.

Sub-Gaussian matrices are often considered to model the measurement process. They have the specific property to be centered, which means that the expected value of each entry is $0$. However, it was shown in a recent work \cite{FlKei} that non-centered matrices have some advantage for the recovery of binary signals. Moreover, in \cite{JK}, a similar phenomenon was observed for the recovery of non-negative signals by \eqref{P+}. To be  more precise the following \emph{biased random matrices} have been considered in \cite{FlKei}:

\begin{align}\label{eqn:biased}
A=\mu\one+D,
\end{align}
where  $\mu\ge 0$ is a freely chosen parameter that controls the expected value of the entries, $\one \in \R^{M,N}$ is the matrix having only entries equal to one, and $D\in \R^{M,N}$ is assumed to be centered and having sub-Gaussian entries whose expected value is $0$.% $d_{i,j}$ with 
%\begin{align}
%\E(d_{i,j})=0 \quad\text{and} \quad\E\left(d_{i,j}^2\right)= \sigma^2.
%\end{align} 

Roughly speaking the following was proven for the recovery of binary, sparse signals from biased random measurements:

\begin{theorem}[Simplified Version of Theorem III.2 and III.8 of \cite{FlKei}]\label{thm:BiasedRand}
	Let $x_0\in \{0,1\}^N$ be a binary vector, and $A\in \R^{M,N}$ be a biased random matrix of the form \eqref{eqn:biased} with $\mu >0$, and $x_0\in \{0,1\}^N$ a $s$-sparse binary vector.
		\renewcommand{\labelenumi}{\roman{enumi})}
	\begin{enumerate}
		\item If $M$ is slightly larger than $N/2$, $x_0$ will be the unique solution of \eqref{Pbin} with high probability.
		\item Under the assumption 
		\begin{align}
		M\gtrsim\max\left(\frac{R^2}{\mu^2},\min(s,N-s)\right)\log(N),
		\end{align} 
		%$x_0$ will be the unique solution of \eqref{Pbin} and of \eqref{LSBin} with high probability. In fact,
		 the solution $x_*$ of \eqref{LSBin} for $y=Ax_0+n$ with $n\in \R^M$ and $\|n\|_2\le \eta$ obeys with high probability
		\begin{align}
		\|x_0- x_*\|_2\le \sqrt{\frac{\left(\frac{\sigma^2}{\mu^2}+32\min(s,N-s)\right)}{m\sigma^2}}\cdot \eta,
		\end{align}
		where $\sigma$ is the variance of the entries of $A$. Particularly, in the case of noiseless measurements, i.e., $\eta=0$, $x_0$ is the unique solution of \eqref{Pbin} and of \eqref{LSBin} with high probability.
	\end{enumerate}
\end{theorem}

\subsection{Main Result}\label{sec:IntMainRes}
As above-mentioned there are several applications where we do not have full freedom to design the measurement matrix. It is therefore of some importance to study structured random matrices. In applications such as radar or wireless communications the measurement process can be represented using partial random circulant matrices and partial random Toeplitz matrices  (cf. \cite{HBRN},\cite{JR}). In those applications binary sparse signals appear also in a natural way. The main goal of this work is therefore to prove comparable results as in \cite{FlKei} (see Theorem \ref{thm:BiasedRand}), for such matrices. Before stating our main results, let us introduce the considered matrices. 

For $b=(b_0,b_1,\dots,b_{N-1})\in \R^N$ we define the associated \emph{circulant matrix} $\Phi=\Phi(b)\in \R^{N,N}$ by setting
\begin{align}
\Phi_{k,j}=b_{j-k\imod N } \qquad k,j\in [N].
\end{align} 
Similarly, for a vector $c=(c_{-N+1},c_{-N+2},\dots, c_{N-1})\in \R^{2N-1}$ the associated \emph{Toeplitz matrix} $T=T(c)\in \R^{N,N}$ has entries 
\begin{align}
T_{k,j}=c_{j-k} \qquad k,j\in [N].
\end{align}
For an arbitrary subset $\Theta\subset[N]$ of  cardinality $M<N$, we let the \emph{partial circulant matrix} $\Phi_{\Theta}=\Phi_{\Theta}(b)\in \R^{M,N}$, and the \emph{partial Toeplitz matrix} $T_{\Theta}$, respectively, be the submatrix of $\Phi$, and $T$ respectively, consisting of the rows indexed by $\Theta$. 
In \cite{StrctRau} one can find a comprehensive overview of compressed sensing with structured random matrices. It is particularly shown that partial circulant matrices with Rademacher input vector $b$ work comparable well for the classical compressed sensing task as completely random sub-Gaussian matrices.

For our purpose we choose the vectors $b$ and $c$ to be sub-Gaussian sequences. Hence, the matrices $\Phi_{\Theta}$ and $T_{\Theta}$ are centered. Similarly to the results in \cite{FlKei}, we consider biased partial random matrices given by

\begin{align}\label{eqn:BiasedCirc}
A=A(b)=\mu \mathds{1}+\Phi_{\Theta}(b), 
\end{align} 
or
\begin{align}\label{eqn:BiasedToep}
A=A(c)=\mu \mathds{1}+T_{\Theta}(c).
\end{align} 
Here, the parameter $\mu\ge0$ controls the expected value of the entries of $A$ and $\one\in \R^{M,N}$ is the matrix having all entries equal to one.

The main purpose of this paper is to prove the symmetric phase transition observed in \cite{FlKei}, for biased partial random matrices given by Equation \eqref{eqn:BiasedCirc}, meaning that we need as many measurements to recover sparse signals as we need to recover dense signals. The main result of this paper is the following theorem:

%\begin{theorem}\label{thm:phaseTrans}
%	Let $b,c\in \R^N$  sub-Gaussian vectors with $\erw{b_{i}^2}=\erw{c_{i}^2}=\sigma^2$, for $i\in [N]$, both with sub-Gaussian norm $R$. Further let $A\in \R^{M,N}$ be a biased measurement matrix of the form \eqref{eqn:BiasedCirc} for some $\mu>0$.  Fix some tolerance $\varepsilon>0$. A binary signal $x_0\in\{0,1\}^N$ with $\|x_0\|_0=s$  is the unique solution of \eqref{Pbin} and \eqref{LSBin} with probability larger than $1-\varepsilon$, provided
%	\begin{align}\label{eqn1:main}
%	M \gtrsim \max\left(\frac{R^2}{\mu^2}, \abs{\min(s, N-s)} \frac{2R^4}{\sigma^4}\right)\log \left(\frac{N}{\varepsilon}\right)
%	\end{align}
%	with a constant depending only on $\sigma$ and $\mu^{-1}$. 
%	
%	Under the additional assumption $M \gtrsim \left(\frac{R}{\sigma}\right)^{4/3}\log(\varepsilon^{-1})$ the solution $x_*$ of \eqref{LSBin} for $b= Ax_0 + n$ with $\norm{n}_2 \leq \eta$ \sg{for some $\eta >0$} obeys
%	\begin{align}\label{eqn:noisyMeas}
%	\norm{x_0-x_*}_2 \leq \sqrt{\frac{9\left(\frac{16\sigma^2}{\mu^2}+ \min(s, N-s)\right)}{M \sigma^2}}\cdot \eta.
%	\end{align}
%\end{theorem}

\begin{theorem}\label{thm:phaseTrans} Let $\mu>0$ and fix some tolerance $\varepsilon>0$. Let $A\in \R^{M,N}$ be a biased measurement matrix 
		\renewcommand{\labelenumi}{\roman{enumi})}
		\begin{enumerate}
		\item  of the form \eqref{eqn:BiasedCirc}, where $b=[b_i]_{i=1}^N\in \R^N$ is a sub-Gaussian vector with $\erw{b_{i}}=0$, $\erw{b_{i}^2}=\sigma^2$, for $i\in [N]$, and sub-Gaussian norm $R$. Or
		\item  of the form \eqref{eqn:BiasedToep}, where $c=[c_i]_{i=-N+1}^{N-1}\in \R^{2N-1}$ is a sub-Gaussian vector with $\erw{c_{i}}=0$, $\erw{c_{i}^2}=\sigma^2$, for $i\in \{-N+1,\dots,N-1\}$, and sub-Gaussian norm $R$.
	\end{enumerate}
    A binary signal $x_0\in\{0,1\}^N$ with $\|x_0\|_0=s$  is the unique solution of \eqref{Pbin} and \eqref{LSBin} for $y=Ax_0$ with  probability larger than $1-\varepsilon$, provided
	\begin{align}\label{eqn1:main}
	M \gtrsim \max\left(\frac{R^2}{\mu^2}, \abs{\min(s, N-s)} \frac{2R^4}{\sigma^4}\right)\log \left(\frac{N}{\varepsilon}\right)
	\end{align}
	with a constant depending only on $\sigma$ and $\mu^{-1}$.

	\begin{enumerate}
		\item[iii)] 	Under the additional assumption $M \gtrsim \left(\frac{R}{\sigma}\right)^{4/3}\log(\varepsilon^{-1})$ the solution $x_*$ of \eqref{LSBin} for $y= Ax_0 + n$ with $\norm{n}_2 \leq \eta$ for some $\eta >0$ obeys
		\begin{align}\label{eqn:noisyMeas}
		\norm{x_0-x_*}_2 \leq \sqrt{\frac{9\left(\frac{16\sigma^2}{\mu^2}+ \min(s, N-s)\right)}{M \sigma^2}}\cdot \eta.
		\end{align}
	\end{enumerate}
\end{theorem}

%% file: Symmetry.tex
\section{Proof of Theorem \ref{thm:phaseTrans}}

We will prove Theorem \ref{thm:phaseTrans} by deriving a so-called \emph{dual certificate} \cite{Fuchs, Gross,TroppDual}, that is a vector $\nu\in \R^M$ having a small $\ell_2$-norm and fulfilling $A^*\nu\in H_S^t$, for some $t\ge0$, where
\begin{align*}
H_S^t:=\left\{w\in \R^N: w_i\le -t\text{ for } i \in S \text{ and } w_i\ge t \text{ for } i \notin S\right\}.
\end{align*}

To justify that this will help to prove the theorem, let us recall some results from \cite{FlKei}.

\begin{proposition}[Propositions II.3 and III.1 of \cite{FlKei}]
		\renewcommand{\labelenumi}{\roman{enumi})}
Let $A\in \R^{M,N}$ and $S\subset[N]$. Then the following statements are equivalent:
	\renewcommand{\labelenumi}{\roman{enumi})}
\begin{enumerate}
	\item $\one_S$ and $\one-\one_S=\one_{S^C}$ are the unique solutions of \eqref{Pbin}.
	\item $\ker(A)\cap H_S^0=\{0\}$.
	\item $\{x\in [0,1]^N: Ax=A\one_S\}=\one_S$.
	\item There is $\nu\in \R^M$ such that $A^*\nu\in H_S^t$ for some $t>0$.
\end{enumerate}
\end{proposition}

Hence, finding a dual certificate indeed ensures, that $\one_S$ is the unique solution of \eqref{Pbin}. However, the third equivalence even yields that there is no other solution of $Ax=A\one_S$ other than $\one_S$ in the box $[0,1]^N$. Thus, the minimization in \eqref{Pbin} is not crucial and we can run box-constrained least-squares \eqref{LSBin} instead.

The proof of Part \emph{iii)} of Theorem \ref{thm:phaseTrans} will further make use of the following result from \cite{FlKei}.
%Another result we aim to recall from \cite{FlKei} will explain Equation \eqref{eqn:noisyMeas}. 

\begin{proposition}\label{prop:noise}[Proposition III.4 of \cite{FlKei}]
Let $r,t, \eta>0$, $S\subset [N]$ and $A\in \R^{M,N}$. Suppose that there exists a dual certificate $\nu\in \R^M$ such that $A^*\nu \in H_S^t$ and $\norm{\nu}_2 \leq r$.

Let $x_0=\mathbb{1}_S\in \R^N$ be the binary signal supported on $S$, and $y=Ax_0 + n$ with $\norm{n}_2 \leq \eta$. Then the solution $x_*$ of the program \eqref{LSBin} obeys
\begin{align*}
\norm{x_*-x_0}_2 \leq \frac{2r}{t}\eta.
\end{align*}	
\end{proposition}

Before constructing the dual certificate explicitly, let us recall a so-called Hoeffding-type inequality as well as Hanson-Wright inequality, which will be an important probabilistic tools for the proof of Theorem \ref{thm:phaseTrans}.

\begin{theorem}[Proposition 5.10 of \cite{VershyninRandomMatrices} and Theorem 1.1 of \cite{rudelson2013hanson}]\label{ProbTools}\
	\renewcommand{\labelenumi}{\roman{enumi})}
	\begin{enumerate}
	\item There exists a universal constant $C>0$ with the following property: If $X_1, \dots, X_M$ are independent sub-Gaussian random variables, then
	\begin{align*}
	\prb{ \bigg\vert \sum_{i=1}^M  X_i\bigg\vert \geq t} \leq e \cdot \exp\left(- \frac{Ct^2}{\gamma^2}\right),
	\end{align*}
	with $\gamma= \sum_{i=1}^M \norm{X_i}_{\psi_2^2}^2$.
	\item There exists a universal constant $C>0$ with the following property: Suppose that $X= (X_1, \dots, X_q) \in \R^{q}$ is a random vector with independent, sub-Gaussian entries. Let further $L$ be a fixed linear map from $\R^q$ to $\R^q$. Then we have
	\begin{align*}
	\prb{ \abs{ \sprod{X, LX} - \erw{\sprod{X,LX}}}>t} \leq 2\exp\left( -C \min \left(\frac{t^2}{R^4\norm{L}_{HS}^2},\frac{t}{R^2\norm{L}}\right) \right)
	\end{align*}
	with $R= \max_{\ell=1, \dots q} \norm{X_\ell}_{\psi_2}$.    
	\end{enumerate}
\end{theorem}

Now we are prepared to prove Theorem \ref{thm:phaseTrans}. Note that we will use the same dual certificate as for the proof of the main theorem of \cite{FlKei} and that the probabilistic tool will also be Hanson-Wright inequality. However, the proof is considerable more sophisticated.

%First note that Toeplitz matrices can be seen as submatrices of circulant matrices, therefore we will only consider circulant matrices within the proof of Theorem \ref{thm:phaseTrans}. In Remark \ref{ProofToepl} we will show how the proof applies to Toeplitz matrices.

\begin{proof}[Proof of Theorem \ref{thm:phaseTrans}]
 
	We first prove \textbf{Part \emph{i)}} of Theorem \ref{thm:phaseTrans}, hence, we assume that the measurement matrix $A$ is a biased partial random circulant matrix.  To make notations easier we enlarge $A$ by inserting zero-rows for indices not in $\Theta$. Hence, we define
	\begin{align}
	A=\mu \one+ \Phi_{\Theta}(b),
	\end{align}
	where 
	\begin{align*}
	(\Phi_{\Theta})_{k,j}=\begin{cases}
	b_{j-k\imod N} & \text{if}\quad k\in \Theta\\
	0 & \text{else}	
	\end{cases} \qquad \text{and}\qquad \one_{k,j}=\begin{cases}
	1 & \text{if}\quad k\in \Theta\\
	0 & \text{else}	
	\end{cases} ,
	\end{align*}
	and $b\in \R^N$ is the given sub-Gaussian vector. This matches the aforementioned measurement process; the vector $Ax_0$ is only enlarged by some zeros. Further we define $\beta_0$ to be the sparser of the two vectors $x_0$ and $\one-x_0$, i.e.,
	\begin{align}
	\beta_0=\begin{cases}
	x_0 & \text{if } \|x_0\|_0\le \|\one-x_0\|_0,\\
	\one-x_0 & \text{else}.
	\end{cases}
	\end{align}
	
	As in \cite{FlKei} we define the dual certificate to be
	\begin{align}
	\nu=\rho \one +\Phi \beta_0-M^{-1}\sprod{\Phi\beta_0,\one}\one, \qquad 	\text{where}\quad \rho=-\frac{\sigma^2}{4\mu},
	\end{align}
	and prove that $A^*\nu\in H_S^t$, where $S=\supp \beta_0$ and $t=\frac{M\sigma^2}{16}$. This particularly means we prove for $i\in [N]$
	\begin{align}\label{eqn:goal}
	\sprod{\nu,A e_i}=\sprod{A^*\nu, e_i}=(A^*\nu)_i\begin{cases}
	\le -t & \text{if } i \in S\\ \ge t  & \text{if } i \notin S.
	\end{cases}
	\end{align}
	A simple calculation yields
	\begin{align}
	\sprod{\nu,A e_i}&=\rho\mu M+\rho \sprod{\one, \Phi_\Theta e_i}+\sprod{\Phi_\Theta\beta_0,\Phi_\Theta e_i}-M^{-1}\sprod{\Phi_\Theta\beta_0,\one}\sprod{\one, \Phi_\Theta e_i}\\&=:\rho\mu M+\rho X_1(i)+X_2(i)-M^{-1}X_3(i).
	\end{align}
	
	We now estimate the numbers $X_1(i), X_2(i), X_3(i)$ for each $i\in S$ and $i\notin S$ separately.
	\newline 
	
	\noindent\textbf{Estimation of $X_1$:}\\
	 We start with $X_1$. For every $i\in [N]$,
	\begin{align}
	X_1(i)=\sprod{\one, \Phi_\Theta e_i}=\sum_{l\in \Theta} (\Phi_\Theta)_{l,i}=\sum_{l\in \Theta} b_{i-l\imod{N}}
	\end{align}
	is a sum of $M$ independent sub-Gaussian variables with sub-Gaussian norm $R$. Thus it follows from Part i) of Theorem \ref{ProbTools} that 
	\begin{align}
	\Pr(\abs{X_1(i)}\ge \theta_1)\le e\exp\left(-\frac{C\theta_1^2}{MR^2}\right),
	\end{align} 
	for every $i\in [N]$.
	The estimations of $X_2$ and $X_3$ are a slightly more involved. Let us start with $X_2$.
	\newline
	
	\noindent \textbf{Estimation of $X_2$:}\\
	 For every $i\in [N]$ it holds true that
	\begin{align*}
	X_2(i)=\sprod{\Phi_\Theta \beta_0, \Phi_\Theta e_i}=\sum_{j=1}^N\sum_{k\in S}(\Phi_\Theta)_{j,k}(\Phi_\Theta)_{j,i}=\sum_{k\in S}\sum_{j\in \Theta}b_{k-j\imod{N}}b_{i-j\imod{N}},
	\end{align*}
	and therefore
	\begin{align}
	\E (X_2(i))=\sum_{k\in S}\sum_{j\in \Theta}\sigma^2\delta_{k-j\imod{N},i-j\imod{N}},
	\end{align}
	where for  $j,k\in[N]$ the number $\delta_{j,k}$ is equal to one for $j=k$ and to zero otherwise. Now it holds true that $k-j\imod{N}=i-j\imod{N}$ if and only if $k=i$ and therefore
	\begin{align}
	\E (X_2(i))=\begin{cases}
	\sum_{j\in\Theta}\sigma^2 \delta_{i-j,i-j} & \text{if } i\in S\\
	0 & \text{if } i \notin S
	\end{cases}\quad=\begin{cases}
	M\sigma^2 & \text{if } i\in S\\
	0 & \text{if } i \notin S.
	\end{cases}
	\end{align}
	
	To estimate the deviation from this expected value we want to use Hanson-Wright inequality. Thus we want to define a map  $L(i):\R^N\to \R^N$ such that $\sprod{b,L(i)b}=X_2(i)$ for all $i\in [N]$. We therefore define the map $L(i):\R^N\to \R^N$, $(v_1,\dots,v_N)\mapsto (w_1,\dots,w_N)$ with
	\begin{align}
	w_j=\begin{cases}
	\sum_{k\in S}v_{k-i+j\imod{N}} & \text{if } j\in \Theta_i\\
	0 & \text{else},
	\end{cases}
	\end{align}
	where $\Theta_i:=i-\Theta$.	Then it indeed holds true that
	\begin{align}
	\sprod{b,L(i)b}=\sum_{j=1}^Nb_j (L(i)b)_j=\sum_{j\in \Theta_i} b_j\sum_{k\in S}b_{k-i+j\imod{N}}=\sum_{j\in \Theta}b_{i-j\imod{N}}\sum_{k\in S}b_{k-j\imod{N}}=X_2(i).
	\end{align}
	Thus, to apply the Hanson-Wright inequality we just need to estimate the Hilbert-Schmidt norm $\|L(i)\|_{HS}$ and the operator norm $\|L(i)\|$. Note that $L$ is a linear map and we therefore can rewrite it as matrix map. The corresponding matrix, which we also call $L(i)=(L_{j,l})_{j,l=1}^N$ is given by
	\begin{align}
	L(i)_{j,l}=\begin{cases}
	1 &\text{if}\; j\in \Theta_i\; \text{and}\; l\in S_{j-i}\\
	0 & \text{else},
	\end{cases}
	\end{align}  
	where $S_{j-i}=j-i+S$. Thus, one easily verifies that $\|L\|_{HS}=M|S|=Ms$. To estimate $\|L\|$ we use the Gershgorin circle Theorem \ref{thm:Gers}. For $j\in \Theta_i$ we have either $L(i)_{j,j}=1$ or $L(i)_{j,j}=0$ and $R_j=s-1$ or $R_j=s$. For $j\notin \Theta_i$ we have $L(i)_{j,j}=0$ and $R_j=0$. Thus by Gershgorin circle theorem all eigenvalues of $L(i)$ lie in the circle $D(0,s)$ and the operator norm can be estimated by $\|L(i)\|\le s$.
	Hanson-Wright inequality therefore implies that there is a universal constant $C>0$ such that for $i\notin S$
	\begin{align}
	\prb{\abs{X_2(i)}>\theta_2}\le 2\exp\left(-C\min\left(\frac{\theta_2^2}{R^4Ms},\frac{\theta_2}{R^2s}\right)\right),
	\end{align}
	and for $i\in S$
	\begin{align}
	\prb{\abs{X_2(i)-M\sigma^2}>\theta_2}\le 2\exp\left(-C\min\left(\frac{\theta_2^2}{R^4Ms},\frac{\theta_2}{R^2s}\right)\right).
	\end{align}
	\newline
	
	\noindent \textbf{Estimation of $X_3$:}\\
	The estimation of $X_3(i)$ is even more involved. We can simplify
	\begin{align*}
	X_3(i)=\sprod{\Phi\beta_0,\one}\sprod{\one, \Phi e_i}=\beta_0^*\Phi^*\one\one^*\Phi e_i=\sum_{k\in S}\sum_{m\in \Theta}\sum_{l\in \Theta}b_{k-m\imod{N}}b_{i-l\imod{N}},
	\end{align*}
	because $(\Phi^*\one\one^*\Phi)_{k,n}=\sum_{m\in \Theta}\sum_{l\in \Theta}b_{k-m\imod{N}}b_{n-l\imod{N}}$. We start by estimating the expected value of $X_3(i)$:
	\begin{align}\label{eqn:E3kn}
	E_{k,n}:&=\E{(\Phi^*\one\one^*\Phi)_{k,n}}=\sum_{m\in \Theta}\sum_{l\in \Theta}\E{\left(b_{k-m\imod{N}}b_{n-l\imod{N}}\right)}\\&=\sum_{m\in \Theta}\sum_{l\in \Theta}\begin{cases} \sigma^2 & \text{if}\; k-m\imod{N}=n-l\imod{N}\\ 0 & \text{else}\end{cases}.
	\end{align}
	Note that $k-m\imod{N}=n-l\imod{N}$ if and only if $k-m-n+l\imod{N}=0$. That is, if and only if $k-m-n+l\in \{-N,0,N\}$, because $k-m-n+l\in (-2N,2N)$. Let us try to identify when this is the case:
	
	First, suppose that $k<n$, then we have $k-n-m<0$, or more precisely either $k-n-m\in[-N,-1]$ or $k-n-m\in[-2N,-N-1]$. If the first case is true,  $l$ can only be chosen such that $k-n-m+l=0$, because $l\in [N]$ there is no possibility that  $k-n-m+l\in \{-N,N\}$. Analogously, if the second case is true, the only possibility to choose $l$ such that  $k-n-m+l\in\{-N,0,N\}$ is $l=m+n-k-N$. In other words for fixed $k<n$ and $m\in[N]$ there is exactly one $l\in [N]$ such that $k-m\imod{N}=n-l\imod{N}$. Similarly we can show that the same is true for $k\ge n$.
	
	However, the sum in \eqref{eqn:E3kn} is not over $m,l\in[N]$ but over the smaller subset $m,l\in \Theta$. Thus the number matching the criterion $k-m\imod{N}=n-l\imod{N}$ is smaller than $M=|\Theta|$, more precisely:
	\begin{align}
	E_{k,n}=\begin{cases}
	\abs{\Theta\cap \left((k-n+\Theta)\cup(N+k-n+\Theta)\right)}\sigma^2 & \text{if}\; k<n,\\
	M\sigma^2 &\text{if}\; k=n,\\
	\abs{\Theta\cap \left((k-n+\Theta )\cup(- N+k-n+\Theta)\right)}\sigma^2 & \text{if}\; k<n.
	\end{cases}=\abs{\Theta\cap(k-n+\Theta)}\sigma^2
	\end{align}
	
	It holds true that for $i\notin S$
	\begin{align}\label{eqn:EX3nS}
	\E(X_3(i))=\sprod{\beta_0,(E_{ki})_{k=1}^N}=\sum_{k\in S} E_{ki}=\sum_{k\in S}\abs{\Theta\cap(k-i+\Theta)}\sigma^2\in [0,sM\sigma^2]
	\end{align}
	and for $i\in S$ that
	\begin{align}
	\E(X_3(i))\in [M\sigma^2,sM\sigma^2].\label{eqn:EX3S}
	\end{align}
%	To justify this statement, let  $\Theta=\{m_1,\dots,m_M\}$, $n\in [N]$ and suppose for $i\in[M]$ that $m_i\in n+\Theta$. Then there is $m_j\in \Theta$ such that $m_i=n+m_j$ which is equivalent to $n=m_i-m_j$. Thus each $m_i\in \Theta$ is contained in the sets $m_j-m_i+\Theta$ for all $j\in [M]$. Thus each $m_i$ is contained in exactly $M-1$ sets of the form $n+\Theta$ for $n\neq 0$ and in $M$ sets if $n$ is also allowed to be zero. This implies $\sum_{n\in \Z, n\neq 0} \abs{\Theta\cap (n+\Theta)}=M(M-1)$ and $\sum_{n\in \Z, n\neq 0} \abs{\Theta\cap (n+\Theta)}=M^2$
%	
%	Hence, for $i\in S$
%	\begin{align}
%	\E(X_3(i)) = \sum_{k\in S}E_{ki}=&\sigma^2\sum_{k\in S,k<i}\abs{\Theta\cap \left((k-i+\Theta)\cup(N+k-i+\Theta)\right)}\\&+\sigma^2\sum_{k\in S,k> i}	\abs{\Theta\cap \left((k-i+\Theta )\cup(- N+k-i+\Theta)\right)}\le M(M-1)\sigma^2.
%	\end{align}
%	And for $i\in S$ it holds $\E(X_3(i)) = \sum_{k\in S}E_{ki}\in [M\sigma^2, M^2\sigma^2]$, because the sum in particular contains $E_{ii}=M\sigma^2$ and thus $\sum_{k\in S}E_{ki}\ge M\sigma^2$.
	
	Now we can compute the probability that $X_3(i)$ deviates from its expected value. We again aim to apply the Hanson-Wright inequality. For this we define $L^3(i):\R^N\to\R^N$, $(v_1,\dots,v_N)\mapsto (w_1,\dots,w_N)$ with
	\begin{align}
	w_l=\begin{cases}
	\sum_{m\in \Theta}\sum_{k\in S} v_{k-m\imod{N}} & \text{if}\; l\in i-\Theta,\\0&\text{else}.
	\end{cases}
	\end{align}
	This yields $X_3(i)=\sprod{b,L^3(i)b}$. To compute the Hilbert-Schmidt norm and operator norm of $L_3(i)$, we further define the matrices $(K^1_{n,m})_{n,m=1}^N$ and $(K^2_{i,j})_{i,j=1}^N$ by
	\begin{align}
	K^1_{n,m}=\begin{cases}
	1 & \text{if}\, n\in i-\Theta , m\in \Theta\\0 & \text{else},
	\end{cases} \quad \text{and} \quad 	K^2_{k,l}=\begin{cases}
	1 & \text{if}\, k\in \Theta , l\in -k+S\\0 & \text{else}.
	\end{cases}
	\end{align}
	It is then easy to verify, that $K^1K^2v=L_3(i)(v)$. Now the Hilbert-Schmidt norm of $K^1$ is given by $\|K^1\|_{\text{HS}}=M^2$ and of $K^2$ by $\|K^2\|_{\text{HS}}=Ms$. Thus
	\begin{align}
	\|L_3(i)\|_{\text{HS}}=\|K^1K^2\|_{\text{HS}}\le \|K^1\|_{\text{HS}}\|K^2\|_{\text{HS}}=M^3s.
	\end{align}
	On the other hand it holds true that
	\begin{align}
		(K^1K^2)_{m,n}=\begin{cases}
		\sum_{l\in \Theta\cap (S-n)}1 &\text{if}\; m\in i-\Theta\\
		0 & \text{else}.
		\end{cases}
	\end{align}
	By  Gershgorin circle theorem this yields  that
	\begin{align}
	\|L_3(i)\|=\|K^1K^2\|\le Ms,
	\end{align}
	since for every $n\in [N]$ it holds true that $\abs{\Theta\cap (S-n)}\le s$ and therefore that the column sum of $K^1K^2$ is smaller than $Ms$ for each column, since the $m$-th entry of each column of $K^1K^2$ is equal to zero if $m\notin i-\Theta$.
	By the Hanson-Wright inequality it follows
		\begin{align}
	\prb{\abs{X_3(i)-\E(X_3(i))}>\theta_3}\le 2\exp\left(-C\min\left(\frac{\theta_3^2}{R^4M^3s},\frac{\theta_3}{R^2Ms}\right)\right).
	\end{align}
	In particular it follows for $i\notin S$
	\begin{align}
	\prb{X_3(i)< -\theta_3}&\le 	\prb{X_3(i)<\E(X_3(i)) -\theta_3}\le \prb{X_3(i)-\E(X_3(i))< -\theta_3}\le	\prb{\abs{X_3(i)-\E(X_3(i))}>\theta_3}\\&\le 2\exp\left(-C\min\left(\frac{\theta_3^2}{R^4M^3s},\frac{\theta_3}{R^2Ms}\right)\right),
	\end{align}
	where we used \eqref{eqn:EX3nS} in the first step. For $i\in S$ it follows
	\begin{align}
	\prb{X_3(i)>\theta_3+M^2\sigma^2}&\le 	\prb{X_3(i)>\theta_3+\E(X_3(i))}\le \prb{X_3(i)-\E(X_3(i))> \theta_3}\\&\le	\prb{\abs{X_3(i)-\E(X_3(i))}>\theta_3}\le 2\exp\left(-C\min\left(\frac{\theta_3^2}{R^4M^3s},\frac{\theta_3}{R^2Ms}\right)\right),
	\end{align}
	where we used \eqref{eqn:EX3S}.
	
	Finally we are able to estimate the probability that Equation \eqref{eqn:goal} is true for $t=\frac{M\sigma^2}{16}$. To this end remember that $\rho=-\frac{\sigma^2}{4\mu}$ and choose
	\begin{align}
	\theta_1 =\frac{\mu M}{4}, \qquad \theta_2 =\frac{\abs{\rho}\mu M}{4}, \qquad \theta_3 =\frac{\abs{\rho}\mu M^2}{4}.
	\end{align}
	
	Using that by Equation \eqref{eqn1:main} we have in particular $M>2s$ we derive for  $i\in S$
	\begin{align}
	\sprod{\nu,A e_i}&=\rho\mu M-\abs{\rho} X_1(i)+X_2(i)-M^{-1}X_3(i)\ge \rho\mu M-\abs{\rho}\theta_1+M\sigma^2-\theta_2-s\sigma^2-M^{-1}\theta_3\\&=\frac{7\rho\mu M}{4}+(M-s)\sigma^2\ge\frac{7\rho\mu M}{4}+\frac{M}{2}\sigma^2 =-\frac{7}{16}M\sigma^2-\frac{1}{2}M\sigma^2=\frac{M\sigma^2}{16}
	\end{align}
	with a failure probability no larger than
	\begin{align}
	s\left(e\exp\left(-\frac{C\mu M}{16R^2}\right)+2\exp\left(-C\min\left(\frac{\rho^2\mu^2M}{16R^4s},\frac{\abs{\rho}\mu M}{4R^2s}\right)\right)+ 2\exp\left(-C\min\left(\frac{\rho^2\mu^2 M}{16R^4s},\frac{\rho\mu M}{4R^2s}\right)\right)\right).
	\end{align}
	
	On the other hand we can estimate for $i\notin S$
	\begin{align}
		\sprod{\nu,A e_i}&\le \rho\mu M +\abs{\rho}\theta_1+\theta_2+M^{-1}\theta_3=\frac{\rho \mu M}{4}=-\frac{M\sigma^2}{16}
	\end{align}
		with a failure probability no larger than
	\begin{align}
	(N-s)\left(e\exp\left(-\frac{C\mu M}{16R^2}\right)+2\exp\left(-C\min\left(\frac{\rho^2\mu^2M}{16R^4s},\frac{\abs{\rho}\mu M}{4R^2s}\right)\right)+ 2\exp\left(-C\min\left(\frac{\rho^2\mu^2 M}{16R^4s},\frac{\rho\mu M}{4R^2s}\right)\right)\right).
	\end{align}
	This finishes the proof of Part \emph{i)} of Theorem \ref{thm:phaseTrans}.
\newline

	Now we aim to prove \textbf{Part \emph{iii)}}. Applying Proposition \ref{prop:noise}, we particularly need to prove the boundedness of the dual certificate $\nu$, i.e., $\|\nu\|_2\le r$, for some $r>0$. First, note that 
	\begin{align}
	\|\nu\|_2^2\le M\rho^2+\sprod{\Phi\beta_0,\Phi\beta_0}.
	\end{align}
	Thus, we in particular need to bound $\sprod{\Phi\beta_0,\Phi\beta_0}$. It is easy to verify that
	\begin{align}
	\sprod{\Phi\beta_0,\Phi\beta_0}=\sum_{n\in \Theta} \sum_{k\in S}\sum_{l\in S}b_{k-n\imod{N}}b_{l-n\imod{N}}.
	\end{align}
	Because $k-n\imod{N}=l-n\imod{N}$ if and only if $k=l$ we obtain 
	\begin{align}
	\E(\sprod{\Phi\beta_0,\Phi\beta_0})=\sum_{n\in \Theta} \sum_{k\in S}\sum_{l\in S}\E(b_{k-n\imod{N}}b_{l-n\imod{N}})=\sum_{n\in \Theta} \sum_{k\in S}\E(b_{k-n\imod{N}}b_{k-n\imod{N}})=Ms\sigma^2.
	\end{align}
	
	To estimate the deviation of $\sprod{\Phi\beta_0,\Phi\beta_0}$ from its expected value we define $L:\R^N\to \R^N$, $(v_1,\dots v_N)\mapsto (w_1,\dots,w_N)$ with
	\begin{align}
	w_i=\begin{cases}
	\sum_{k\in S}v_{k-i\imod{N}} & \text{if } i\in \Theta\\ 0 &\text{else}.
	\end{cases}
	\end{align} 
	Then it holds true that $\sprod{\Phi\beta_0,\Phi\beta_0}=\sprod{Lb,Lb}:=\sprod{b,Kb}$, with $K:=L^*L$. To again apply the Hanson-Wright inequality we need to estimate the Hilbert-Schmidt norm $\|K\|_{\text{HS}}$ as well as the operator norm $\|K\|$ of $K$. Note that $L$ can be represented by the matrix $[L_{ij}]_{i,j=1}^N$, where
	\begin{align}
	L_{ij}=\begin{cases}
	1 & \text{if } i\in \Theta \text{ and } j\in S-i\\
	0& \text{else}
	\end{cases}=\begin{cases}
	1 & \text{if } i\in \Theta\cap S-j\\
	0& \text{else}
	\end{cases}.
	\end{align}

	Further $K$ can be represented by the matrix $[K_{kl}]_{k,l=1}^N$ with
	\begin{align}
	K_{kl}=\sum_{i=1}^N L_{ki}^*L_{il}=\sum_{i\in \Theta\cap(S-k)\cap (S-l)}1=\abs{\Theta\cap(S-k)\cap (S-l)}.
	\end{align}
	By Gershgorin circle theorem all eigenvalues of $K$ lie in the circle $D(0,\sum_{k\in[N]}K_{kl})$ for all $l\in [N]$ and 
	\begin{align}
	\sum_{k\in[N]}K_{kl}=\sum_{k\in[N]}\abs{\Theta\cap(S-k)\cap (S-l)}\le\sum_{k\in[N]}\abs{(S-k)\cap (S-l)} \le s^2,
	\end{align}
	where the last step follows from the following fact: Let $l\in [N]$ and $S-k=\{k_1,\dots, k_s\}$ then for $i\in [s]$ it holds true that $k_i\in K-(k_j-k_i)$ for each $j\in [s]$. Thus, each $k_i\in S-k$ is contained in exactly $s$ sets of the form $S-r$.
	Therefore, we can conclude $\|K\|\le s^2$ and $\|L\|=\sqrt{\lambda_{\text{max}}(L^*L)}=\sqrt{\lambda_{\text{max}}(K)}=s$, where $\lambda_{\text{max}}$ denotes the largest eigenvalue. Further it is easy to verify by the matrix representation that $\|L\|_{\text{HS}}=\sqrt{Ms}$ and therefore by \eqref{eqn:EstHS}
	\begin{align}
	\|K\|^2_{\text{HS}}=\|L^*L\|^2_{\text{HS}}\le \|L^*\|^2\|L\|^2_{\text{HS}}\le Ms^3.
	\end{align}

	By the Hanson-Wright inequality it now follows
	\begin{align}
	\prb{\abs{\sprod{\Phi\beta_0,\Phi\beta_0}-Ms\sigma^2}\ge Ms\sigma^2 }\le2 \exp\left(-c\min\left(\frac{M\sigma^4}{sR^4},\frac{M\sigma^2}{R^2s}\right)\right).
	\end{align}
	Thus $\sprod{\Phi\beta_0,\Phi\beta_0}\le 2Ms\sigma^2$ with high probability. And more precisely we derive for the dual certificate that
	$\|\nu\|_2^2\le M\sigma^2\left(\frac{\sigma^2}{16\mu^2}+2s\right)$ with probability $1-\varepsilon$ if $M\gtrsim \max\left(\frac{R^4}{\sigma^4},\frac{R^2}{\sigma^2}\right)\ln\left(\frac{2}{\varepsilon}\right)s$ for some $\varepsilon>0$.
	\newline
	
	It remains to prove \textbf{Part \emph{ii)}}. For this purpose it is enough to argue that the previous proof also applies to Toeplitz matrices of the form \eqref{eqn:BiasedCirc}. To see this, note that Toeplitz matrices are submatrices of circulant matrices. More precisely, let  $T=T(c_{-N+1},c_{-N+2},\dots, c_{N-1})\in \R^{N,N}$ then T is the submatrix of $\Phi([c_0,\dots,c_{N-1},c_{-N+1},\dots,c_{-1}])\in \R^{2N-1,2N-1}$ consisting of the first $N$ columns and rows of $\Phi$, i.e., $T_{i,j}=C_{i,j}$ for $i,j\in [N]$. 

Define $A=\mu\one+T_{\Theta}$ and $B=\mu\one+\Phi_{\Theta}$. For the dual certificate $\nu$ and $i\in [N]$ it then holds true that $\sprod{\nu,A e_i}=\sprod{\nu,B\tilde{e}_i}$, where $\tilde{e}_i$ denotes the $i$-th canonical vector in $\R^{2N-1}$. Further note that $Ax=B[x \quad\vec{0}]$ for $x\in \R^N$, where $\vec{0}\in \R^{N-1}$. 

Thus the former proof shows that $B^*\nu\in \tilde{H}_S^t:=\left\{w\in \R^{2N-1}: w_i\le -t\text{ for } i \in S \text{ and } w_i\ge t \text{ for } i\in [N]\setminus S\right\}$ for some $t>0$. Note that for $w\in\tilde{H}_S^t$, the last $N-1$ entries of $w$ can be arbitrary. Moreover it is easy to prove similarly as in \cite{FlKei} (cf. Proposition 2.3 in \cite{FlKei}) that this is equivalent to $$\left\{x\in [0, 1]^{2N-1}: Bx=B\one_S\right\}\cap \left\{x\in[0,1]^{2N-1}: x_{i}=0, i=N+1,\dots,2N-1 \right\}=\one_S\in \R^{2N-1}.$$  Hence, $\left\{x\in [0, 1]^{N}: Ax=A\one_S\right\}=\one_S\in \R^{N}$, under same assumption of Theorem \ref{thm:phaseTrans}.
\end{proof}

\begin{remark}
	Note that the (implicit) constant in Equation \eqref{eqn1:main} is doubled in comparison to the result in \cite{FlKei} for non structured matrices. The main reason is the pessimistic estimation of the expected value of $X_3(i)$ in Equation \eqref{eqn:EX3S}. Particularly, $\E{X_3(i)}=sM\sigma^2$ can only be true for very artificial choices of $\Theta$ and $S$. Suppose $\Theta=\{m,m+k,m+2k,\dots,m+(M-1)k\}$ such that $Mk=N$ and $S_i=\{0,k,2k,3k,(s-1)k\}$ then  $\E{X_3(i)}=sM\sigma^2$. However, if there is no $k\in S$ such that $\Theta=\{m,m+k,m+2k,\dots,m+(M-1)k\}$ then $\abs{\Theta\cap (k+\Theta)}\le (M-1)$. Moreover, if $S_i\neq \{0,k,2k,3k,(s-1)k\}$ for some $k$, then for $l\in S_i, l\neq k$, it holds true that $\abs{\Theta\cap (l+\Theta)}\le (M-1)$.
	
	One possibility to ensure that $\E(X_3(i))\le\frac{Ms}{2}$, which gives the same constant as in \cite{FlKei}, is to choose $\Theta=\{m_1>m_2>\dots m_M\}$ such that every possible distance between $m_i$ and $m_{i+1}$ arises at most $M/2$-times. This might be possible because we may assume that $M\le \frac{N}{2}$ by what the numerics indicate (see Section \ref{sec:numeric}). 
\end{remark}

%	
%As mentioned in the beginning we assumed in the proof of Theorem \ref{thm:phaseTrans} that the measurement matrix $A$ is a circulant matrix of the form \eqref{eqn:BiasedCirc}. Thus it remains to argue that the proof also applies to Toeplitz matrices of the form \eqref{eqn:BiasedCirc}.

%% file: Numerics.tex
\section{Numerical Validation}\label{sec:numeric}

To support our theory we aim to conclude by showing the results of the following numerical experiments. Basically we run the boxed-constrained basis pursuit \eqref{Pbin} and the box-constrained least squares \eqref{LSBin} for biased circulant matrices $\Phi(b)$ and Toeplitz matrices $T(b)$ for either Gaussian or Rademacher input vectors. For all experiments we choose the ambient dimension to be $N=500$. For each $s,M\in \left\{5i:i\in [100]\right\}$ we choose a random binary vector $x_0\in\R^N$ with $s$-non-zero elements. These $s$-non-zero positions of $x_0$ are chosen uniformly randomly from $1$ to $500$ without repetition (Matlab method \textit{randperm}). Then we constructed the different structured measurement matrices and run boxed constrained linear least squares (Matlab method \textit{lsqlin}) as well as the linear program \textit{linprog} with box constraints to obtain reconstructions $x_{\text{LS}}$ and $x_{\text{PBin}}$ of $x_0$. Finally we computed the $\ell_2$-error $\|x_{\text{LS}}-x_0\|_2$ and $\|x_{\text{PBin}}-x_0\|_2$ and repeated the computation for each combination of $s$ and $M$ $100$-times.

The measurement matrices were constructed as follows. For the Rademacher Toeplitz matrix we choose  $c=[c_1,\dots, c_N]\in \R^N$ and $r=[r_1,\dots,r_{N-1}]\in \R^{N-1}$ such that $c_i,r_j$ are either $0$ or $1$ with equal probability for $i\in[N]$, $j\in [N-1]$. Then we define $A=\text{toeplitz}(c,[c_1\; r])$, thus $c$ is the first column of $A$ and $[c_1\; r]$ the first row. Finally we choose a random subset  $\Theta\subset[N]$ of size $\abs{\Theta}=M$ by randomly permutating $[N]$ and choosing the first $M$ numbers (matlab method \textit{randperm}$(N,M)$). The measurement matrix $\Phi$ then consist of the columns $A$ which correspond to the subset $\Theta$. 

 The Gaussian Toeplitz matrix is constructed analogously but with $c\in \R^N$ and $r\in \R^{N-1}$ shifted Gaussian random vectors, i.e, $c=g_c+\one$ and $r=g_r+\one$, where $g_c\in \R^N$ and $g_r\in\R^{N-1}$ are Gaussian vectors.
 
To construct the partial circulant matrices we choose $c\in \R^N$ as in the Toeplitz case either as shifted Rademacher or Gaussian vector. Then we flipped $c$ and shifted it circularly by one position to obtain the row vector $r\in \R^N$ ($r=$\textit{circshift(fliplr}($c$),$1,2$)). Then we define $A=\text{toeplitz}(c,r)$ and the measurement matrix $\Phi$ as in the Toeplitz case.

\begin{figure}[H]
	\begin{center}
		\includegraphics[trim=40 200 50 200,clip, width=0.45\textwidth]{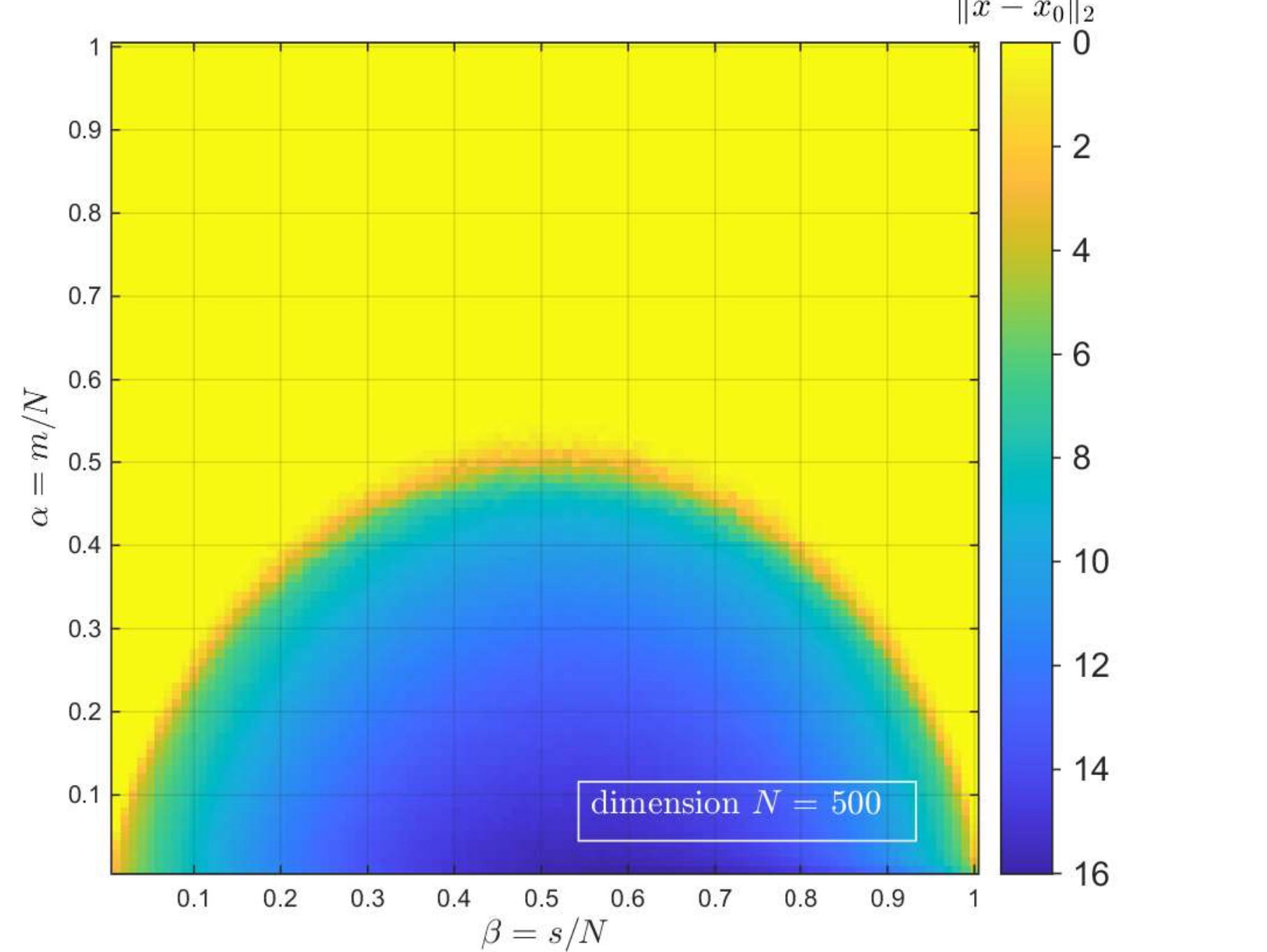}
		\includegraphics[trim=40 200 50 200,clip, width=0.45\textwidth]{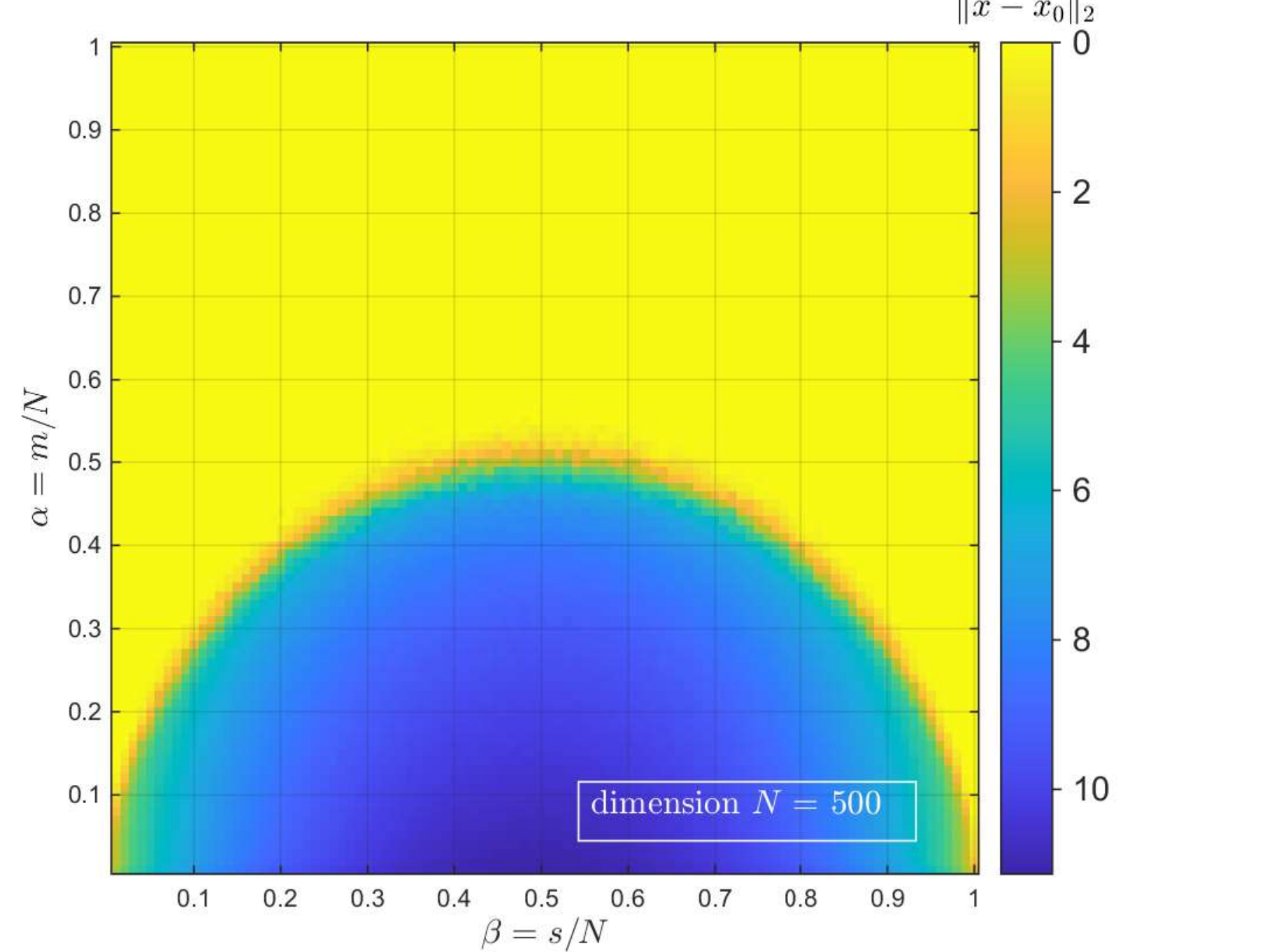}
		\caption{Reconstruction from biased Rademacher Toeplitz measurements via \eqref{Pbin} (left) and \eqref{LSBin} (right). The experiment yielding this figure is explained above.}
		\label{ToeplitzBernoulli}
	\end{center}
\end{figure}
\begin{figure}[H]
	\begin{center}
		\includegraphics[trim=40 200 50 200,clip, width=0.45\textwidth]{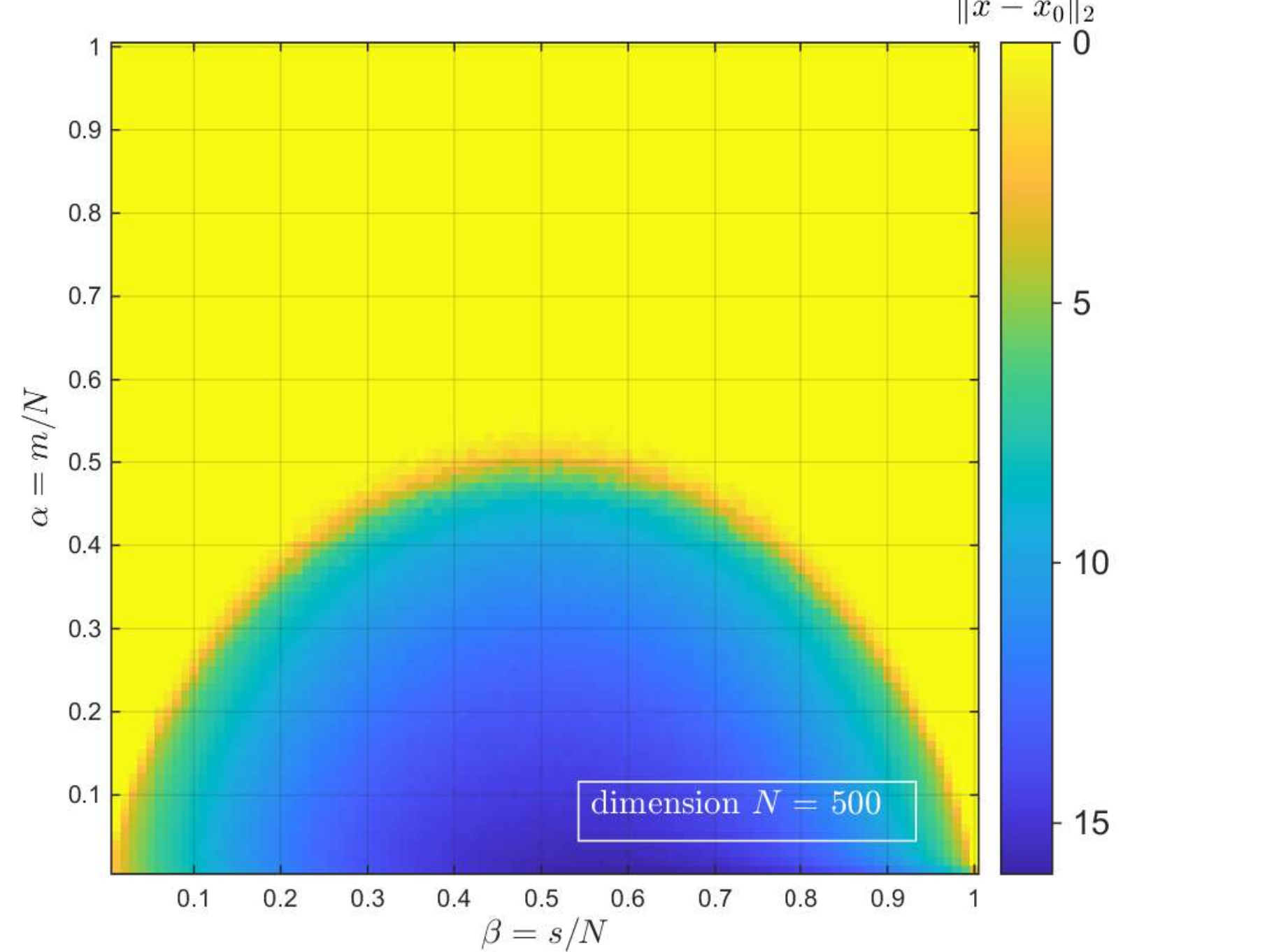}
		\includegraphics[trim=40 200 50 200,clip, width=0.45\textwidth]{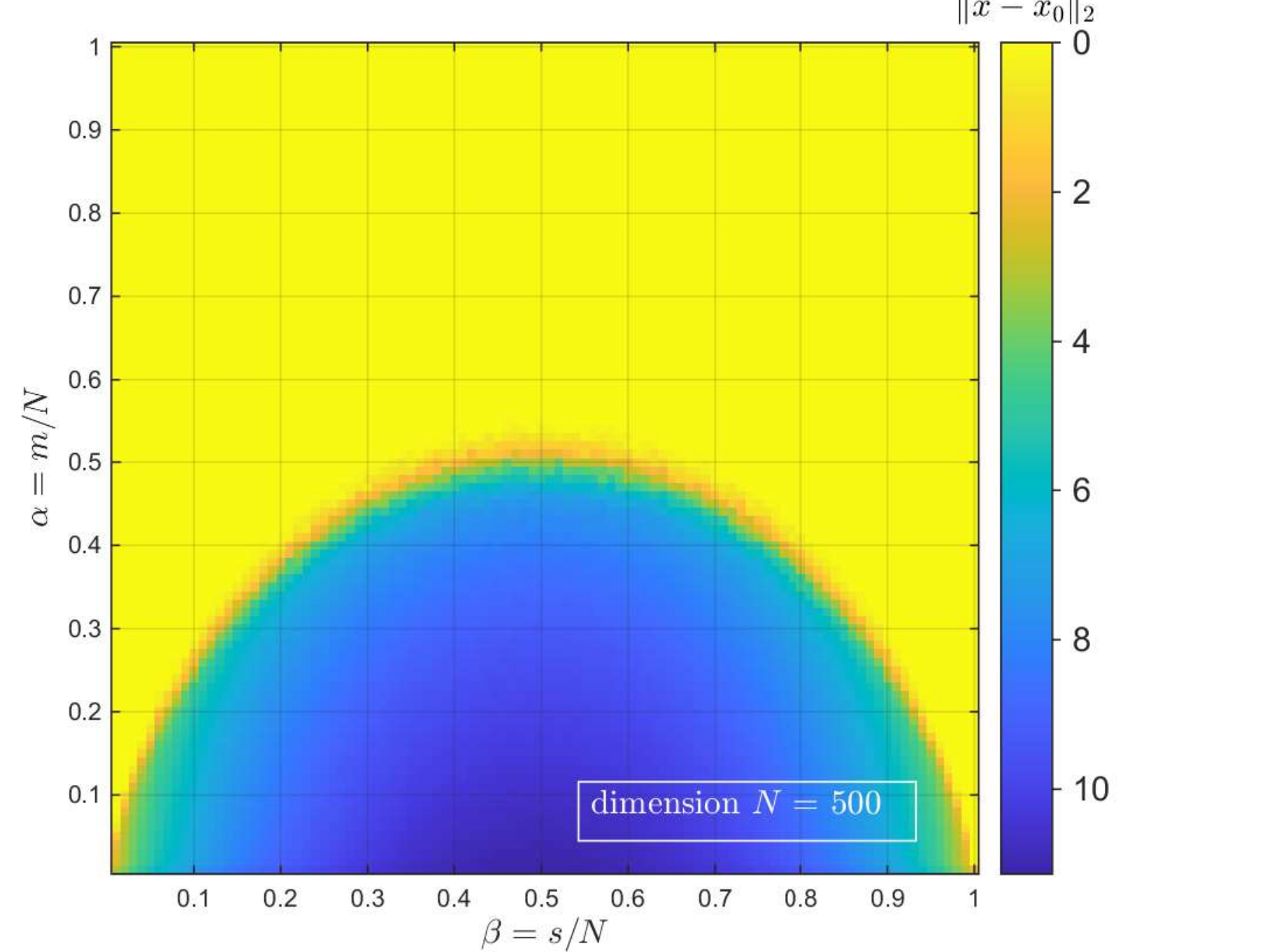}
		\caption{Reconstruction from biased Gaussian Toeplitz measurements via \eqref{Pbin} (left) and \eqref{LSBin} (right). The experiment yielding this figure is explained above.}
		\label{ToeplitzGaussian}
	\end{center}
\end{figure}
\begin{figure}[H]
	\begin{center}
		\includegraphics[trim=40 200 50 200,clip, width=0.45\textwidth]{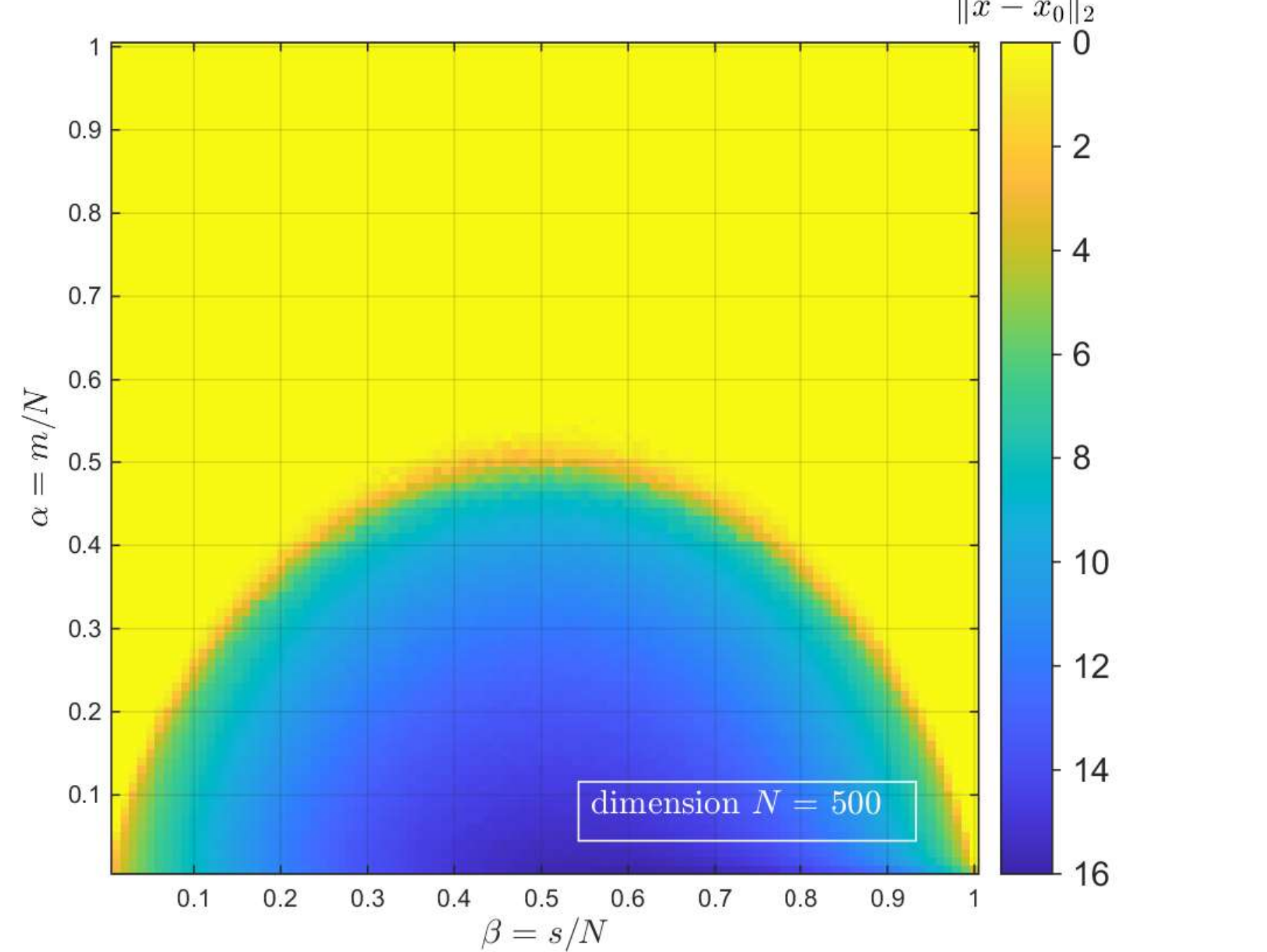}
		\includegraphics[trim=40 200 50 200,clip, width=0.45\textwidth]{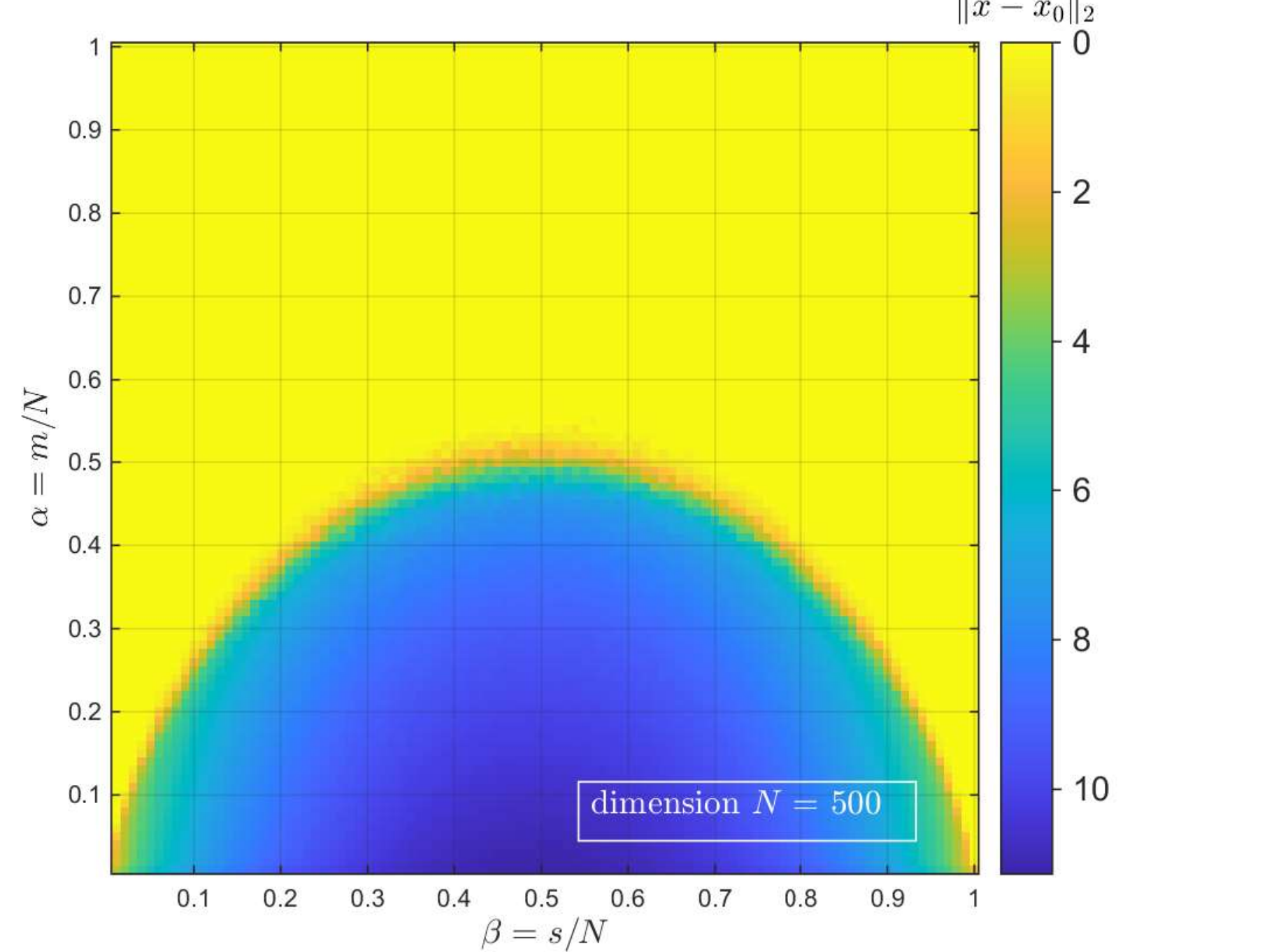}
		\caption{Reconstruction from biased Rademacher circulant measurements via \eqref{Pbin} (left) and \eqref{LSBin} (right). The experiment yielding this figure is explained above.}
		\label{CirculantBernoulli}
	\end{center}
\end{figure}
\begin{figure}[H]
	\begin{center}
		\includegraphics[trim=40 200 50 200,clip, width=0.45\textwidth]{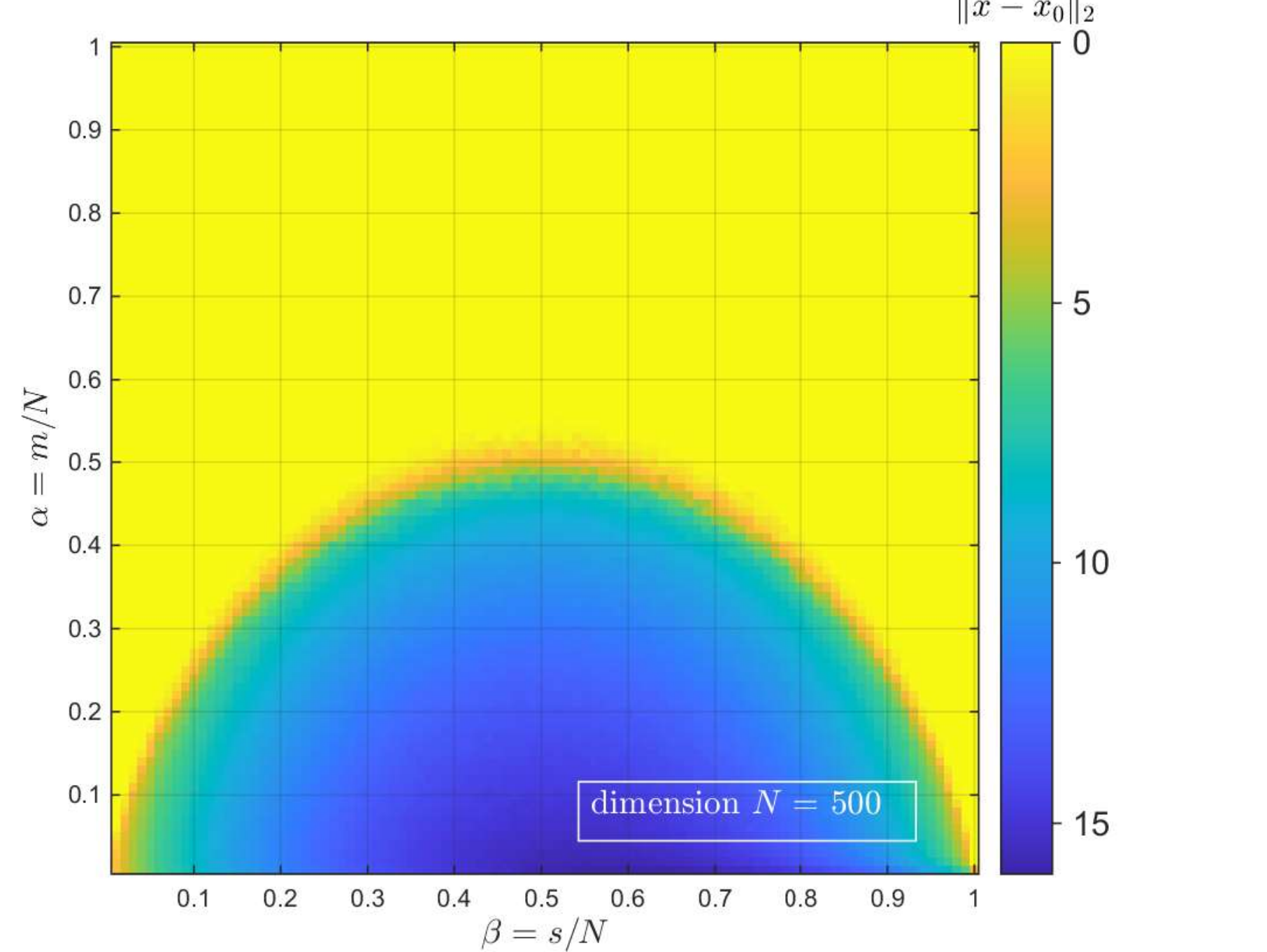}
		\includegraphics[trim=40 200 50 200,clip, width=0.45\textwidth]{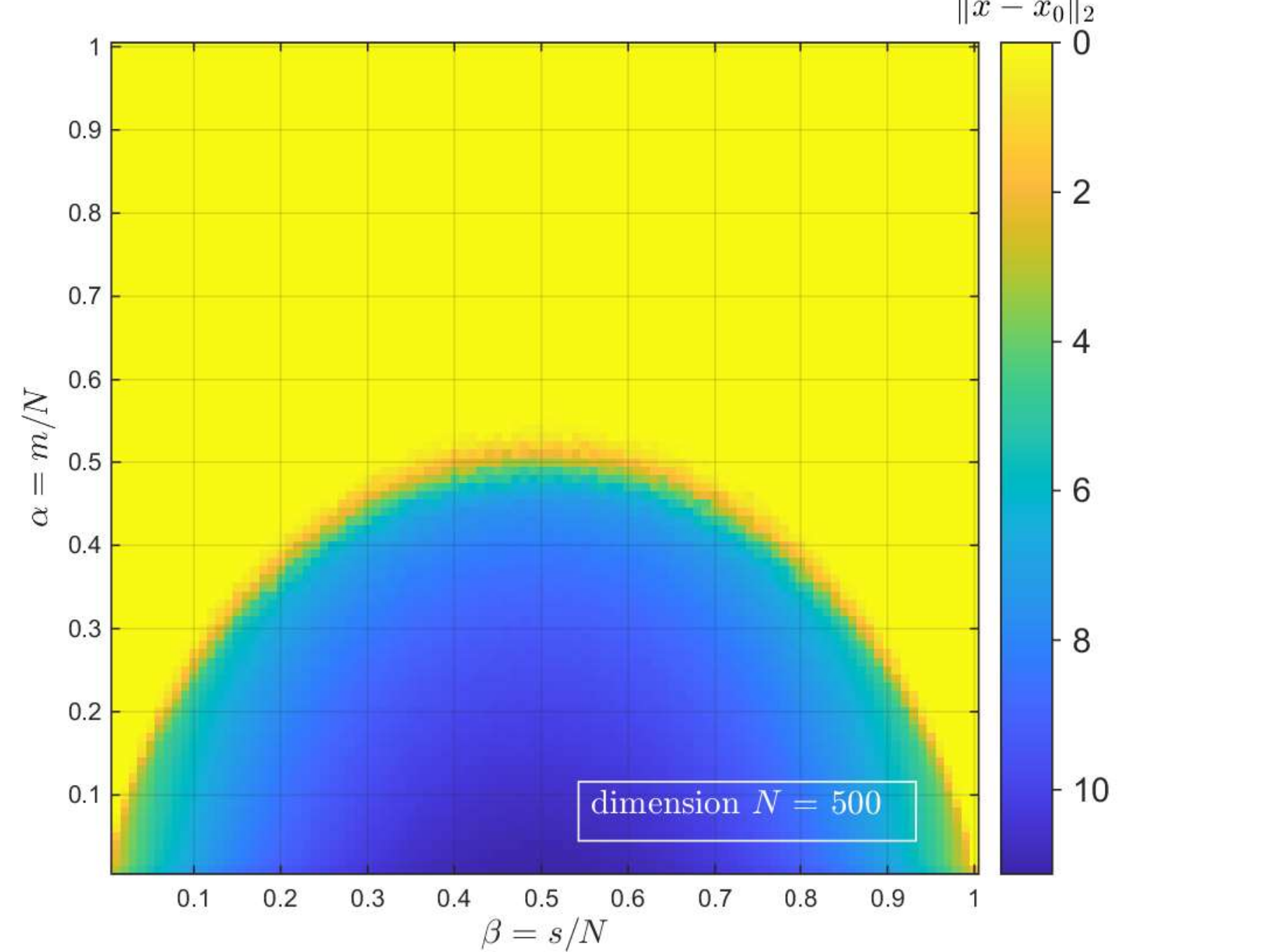}
		\caption{Reconstruction from biased Gaussian circulant measurements via \eqref{Pbin} (left) and \eqref{LSBin} (right). The experiment yielding this figure is explained above.}
		\label{CirculantGaussian}
	\end{center}
\end{figure}

%% file: UpperBound.tex
\section{Outlook and Future Work}
In this paper we have studied the reconstruction of a binary, sparse signal $x_0\in \{0,1\}^N$ from biased partial random circulant and Toeplitz measurements $y=Ax_0 +n$, where $n$ is some small noise vector and $A=\Phi(b)+\mu\one$. Here, $\Phi(b)$ a partial random circulant or Toeplitz matrix and $b$ is a sub-Gaussian vector.
In particular, we have studied the reconstruction via basis pursuit with box constraints and linear least squares with box constraints. Surprisingly, we could prove that the least squares algorithm, which usually does not promote sparsity, works comparably well. We further showed that we need as many measurements to recover an $s$-sparse signal as we need to recover an $(N-s)$-sparse signal and that this number is about $\min\{s,N-s\}\log{N}$. Finally we substantiated our theory by some numerical simulations.

The numerical simulations further indicate that, as for non-structured sub-Gaussian measurements, the number of necessary measurements $M$ to ensure unique recovery (independently of the sparsity) is not larger than $M=N/2$. In the following we would like to explain some thoughts on this phenomenon and a possible proof of it.

In order to prove the mentioned phenomenon, we would need to show that the measurement matrix $\Phi_{\Theta}$ is in general position and orthant symmetric; see the proof of Theorem III.2 in \cite{FlKei}, for comparison.

Note, that $\Phi_{\Theta}$ is in general position if for every subset $B\subset[N]$ of cardinality $|B|=M$ the eigenvalues of the matrix $(\Phi_\Theta^{B})^*\Phi_\Theta^B$ are positive, i.e., if $\lambda(\Phi_\Theta^{B})^*\Phi_\Theta^B)>0$ or $s_N(\Phi_\Theta^B)\ge0$, respectively, where $s_N(A)$ denotes the smallest singular value of a matrix $A\in \R^{N,N}$. Note that there exists a lot of work in the literature calculating the probability of $\lambda(\Phi^*\Phi)>0$ (\cite{Paulo, Gray}), but an analogous calculation for $\lambda((\Phi^{B})^*\Phi^B)$ is much more difficult. The main reason is that the eigenvalues of $\Phi$ itself can be computed very easily, which is not the case for a submatrix. In particular it has been shown in \cite{Paulo}, for Rademacher sequences $(b_0,\dots,b_{N-1})$ (among others) that for all $\varepsilon >0$ and large $N$
\begin{align}\label{eqn:EVCirc}
\prb{s_N(\Phi(b))\ge \varepsilon N^{-1}}\ge 1-C\varepsilon,
\end{align}
where $C>0$ is a constant only depending on $b$.
Since $\Phi^B_\Theta$ seems to be more unstructured than $\Phi$ itself one might hope that the probability of $\lambda((\Phi^{B}_\Theta)^*\Phi^B_{\Theta})>0$ is even higher.

However, the proof of Theorem III.2 in \cite{FlKei}, is also based upon the fact that the random part $\Phi$ of the measurement matrix $A=\mu \one +\Phi$ has an orthant symmetric kernel. This means, if for each diagonal matrix $S\in \R^{N,N}$ with diagonal in $\{-1,1\}^N$ and every measurable set $\Omega\in \R^{M,N}$, it holds $\prb{B S\in \Omega}=\prb{B\in \Omega}$, where $B$ is a matrix whose rows span the kernel of $\Phi$. 

From our point of view, it seems to be very unlikely that this is true for partial circulant matrices. The reason is the following. Suppose the kernel of $\Phi_{\Theta}$ is spanned by $v_1,\dots,v_m$, for some $m\in [N]$, and let $B\in \R^{m,N}$ the matrix with rows $v_1',\dots,v_m'$, then the rows of $BS$ span the kernel of $TS$. But $TS$ is in general no circulant marix. 

So for future work it is interesting to check if the kernel of a partial circulant matrix is indeed not or perhaps is orthant symmetric. Independently of the answer to this question it might be of own interest to prove an inequality in the spirit of \eqref{eqn:EVCirc} for partial circulant matrices. And, if the answer to the first question is negative, to find an alternative way to proof the upper bound on the necessary number of measurements in the order of $M\lesssim N/2$.